\title{Repeated Communication with Private Lying Cost}
\author{Harry PEI\footnote{Department of Economics, Northwestern University. I thank  Drew Fudenberg,
Piotr Dworczak,
 Yingni Guo, Delong Meng, Wojciech Olszewski, Teck Yong Tan, Eran Shmaya, and Bruno Strulovici for helpful comments.}}
\date{First Draft: January 28, 2020. This Draft: \today}

\documentclass[11pt]{article}
\usepackage{amsmath}
\usepackage{graphicx}
\usepackage{amsfonts}
\usepackage{amsthm}
\usepackage{setspace}
\usepackage{tikz}
\usetikzlibrary{patterns}
\usepackage{sgame}
\usepackage{color}
\usepackage{hyperref}
\usepackage[top=1in, bottom=1in, left=1in, right=1in]{geometry}
\begin{document}
\newtheorem{Proposition}{\hskip\parindent\bf{Proposition}}
\newtheorem{Theorem}{\hskip\parindent\bf{Theorem}}
\newtheorem*{Theorem1}{\hskip\parindent\bf{Theorem 1'}}
\newtheorem*{Theorem2}{\hskip\parindent\bf{Theorem 2'}}
\newtheorem*{Assumption2}{\hskip\parindent\bf{Assumption 2'}}
\newtheorem{Lemma}{\hskip\parindent\bf{Lemma}}[section]
\newtheorem{Corollary}{\hskip\parindent\bf{Corollary}}
\newtheorem{Definition}{\hskip\parindent\bf{Definition}}
\newtheorem{Assumption}{\hskip\parindent\bf{Assumption}}
\newtheorem{Condition}{\hskip\parindent\bf{Condition}}
\newtheorem{Conjecture}{\hskip\parindent\bf{Conjecture}}
\newtheorem*{Lemma1}{\hskip\parindent\bf{Lemma A.0}}

\maketitle
\numberwithin{equation}{section}

\noindent I study repeated communication games between a patient sender and a sequence of receivers. The sender has \textit{persistent private information} about his psychological cost of lying, and in every period, can privately observe the realization of an i.i.d. state before communication takes place. I characterize every type of sender's highest equilibrium payoff. When the highest lying cost in the support of the receivers' prior belief approaches the sender's benefit from lying, every type's highest equilibrium payoff in the repeated communication game converges to his equilibrium payoff in a one-shot Bayesian persuasion game. I also show that in every sender-optimal equilibrium, no type of sender mixes between telling the truth and lying at every history. When there exist ethical types whose lying costs outweigh their benefits, I provide necessary and sufficient conditions for all non-ethical type senders to attain their optimal commitment payoffs. I identify an \textit{outside option effect} through which the possibility of being ethical decreases every non-ethical type's payoff.\\


\begin{spacing}{1.5}

\section{Introduction}\label{sec1}
Economists have long recognized that
informed experts' commitment power
has significant effects
on communication outcomes.
In the seminal work of Crawford and Sobel (1982),
an expert's temptation to
mislead his audience undermines his credibility, which leads to ineffective information transmission and low social welfare.
When the expert can \textit{commit} to disclosure policies, as in the model of Kamenica and Gentzkow (2011),
his messages can have more influence over others' decisions, and his payoff improves relative to the benchmark scenario without commitment.

In practice, experts face credibility issues when committing to disclosure policies.\footnote{Two exceptions to this statement include: the leading application of Kamenica and Gentzkow (2011), in which a prosecutor is required by law to disclose everything he learns from investigations, as well as pharmaceutical companies that are legally obliged to disclose the outcomes of drug trials to the FDA.}
This is especially the case when an expert has private interests on
his advisees' decisions and his optimal disclosure policy is stochastic.
In these scenarios, it is against the expert's own interest to honor his commitment after
receiving payoff-relevant private information.

A plausible microfoundation for the expert's commitment is that he communicates with multiple receivers one at a time, and each receiver observes the expert's past recommendations and compare them with past state realizations.
However, according to the results of Fudenberg, Kreps and Maskin (1990) and Fudenberg and Levine (1994),
when the expert's optimal disclosure policy is nontrivially stochastic,
his highest equilibrium payoff in the repeated game is strictly bounded below his optimal commitment payoff no matter how patient he is.
This is because
receivers \textit{cannot} perfectly monitor the expert's stochastic disclosure policies, and as a result,
inefficient punishments need to occur on the equilibrium path in order to provide the expert incentives.

I examine the extent to which a strategic expert can restore his commitment power in repeated communication games
when he has \textit{persistent private information} about his \textit{psychological cost of lying}.
In my model, a patient sender communicates with an infinite sequence of receivers, arriving one in each period and each plays the game only in the period she arrives.
The stage-game follows from the leading example in Kamenica and Gentzkow (2011).
In every period, the sender privately observes the realization of an i.i.d. state, which is either \textit{high} or \textit{low}, and
recommends either a \textit{high action} or a \textit{low action}. The receiver chooses one of the two actions after observing the sender's recommendation together with the history of states and recommendations.
The receiver wants to match her action with the state.
The sender strictly prefers the high action regardless of the state, and suffers from
a psychological cost of lying
when
his recommendation fails to match the state.

I assume each receiver strictly prefers the low action under her prior belief, and first focus on settings in which the sender's highest possible lying cost is strictly lower than his benefit from the high action. The latter implies that \textit{all} types of the sender are \textit{non-ethical}, in the sense that they have \textit{strict incentives} to mislead a receiver who takes their messages at face value. This modeling assumption contrasts to the commitment-type model of Mathevet, Pearce and Stacchetti (2019), in which the sender can commit to disclosure policies with positive probability.

Theorem \ref{Theorem1} characterizes every type of patient sender's \textit{highest equilibrium payoff}. To compute this payoff, consider an auxiliary \textit{static optimization problem}
in which a planner chooses a distribution over \textit{stage-game action profiles} to maximize this type of sender's expected payoff subject to two constraints.\footnote{Each stage-game action for the sender is a mapping from the set of states to the set of messages. Each stage-game action for the receiver is a mapping from the set of messages to the set of actions. A stage-game action profile consists of a stage-game action for the sender and a stage-game action for the receiver.}
First, the highest-cost type sender's payoff under this distribution is no more than his payoff under the full disclosure policy.
Second, the receiver's stage-game action best replies against the conditional distribution over the sender's stage-game actions.

My characterization result has three implications. First, every type's highest equilibrium payoff depends only on his true cost of lying and the highest cost of lying in the support of receivers' prior belief. Second, for every type who does not have the highest cost, his highest equilibrium payoff is \textit{strictly greater} than his highest attainable payoff in a repeated game where his lying cost is common knowledge.
This is somewhat puzzling given that the sender needs to extract information rent (or equivalently, reveal information about his persistent type) in order to obtain such high payoffs.
As the sender becomes more patient, the number of periods
he needs to extract information rent to attain a given discounted average payoff
\textit{grows without bound}. How can he reveal persistent private information
for unboundedly many times while preserving his informational advantage?
Third, as the highest lying cost in the support of receivers' prior belief converges to the sender's benefit from lying, every type of sender's highest equilibrium payoff in the repeated game converges to his optimal commitment payoff in the one-shot game.
This observation provides a strategic justification for the sender's commitment to disclosure policies in Bayesian persuasion models.

My second result clarifies the distinction between a strategic-type sender who faces high lying cost and a commitment-type sender who uses his optimal disclosure policy in every period. Theorem \ref{Theorem2} shows that whenever the sender has two or more types, in every sender-optimal equilibrium, no type mixes between telling the truth and lying (i.e., recommending the high action in both states) at all on-path histories. It then implies that no type adopts his optimal disclosure policy at every on-path history in \textit{any equilibrium} (no matter whether it is sender-optimal or not).
These conclusions extend to a type
 whose lying cost exactly offsets his benefit from the high action.

To understand why, suppose toward a contradiction that one type of sender mixes between the two stage-game actions at every history. Then telling the truth at every history and lying at every history are both his best replies. First, suppose this type is not the one with the highest lying cost, then
the highest-cost type tells the truth with probability one at every on-path history. As a result, the second-highest-cost type will be separated from the highest-cost type as soon as he lies, after which
his cost becomes the highest
in the support of receivers' posterior belief. This implies that his equilibrium payoff cannot exceed his payoff in the repeated complete information game, which contradicts the second implication of Theorem \ref{Theorem1} that the second-highest-cost type strictly benefits from incomplete information.
Second, suppose this type is the one with the highest lying cost, then lying at every history is one of his equilibrium best replies, from which he obtains his highest equilibrium payoff in the repeated incomplete information game. This suggests a \textit{lower bound} on the second-highest-cost type's  payoff by lying in every period, which one can verify that it is
strictly greater than his highest equilibrium payoff. This leads to a contradiction.

My third result incorporates the possibility that the sender is \textit{ethical} in the sense that his lying cost outweighs his benefit from the high action. Theorem \ref{Theorem3} provides a necessary and sufficient condition for all non-ethical types to attain their optimal commitment payoffs. My condition depends only on the \textit{highest} and the \textit{lowest lying costs among the ethical types}. Moreover, given the existence of at least one ethical type, and fixing one of the aforementioned variables, this condition is satisfied
if and only if the other variable \textit{falls below} a cutoff.
This result suggests that a non-ethical sender can be \textit{worse off} when receivers entertain the possibility
of an ethical-type with a high lying cost. This contrasts to situations without ethical types in which
every type of patient sender's highest equilibrium payoff strictly increases with
the highest  lying cost.

The above observation is driven by an \textit{outside option effect} that is absent in models without ethical types as well as reputation models with commitment types.
When receivers entertain the possibility
of an ethical type with a high lying cost, the ethical types whose lying costs are relatively low enjoy better outside options given that they can imitate the equilibrium strategy of this newly introduced high-cost type.
Such an improvement in outside options limits the frequency with which each ethical type can lie, and as a result, reduces non-ethical types' opportunities to lie while pooling with at least some ethical types.

I construct equilibria that exhibit \textit{slow learning} and \textit{reputation rebuilding} to establish the attainability of high payoffs in
 Theorems \ref{Theorem1} and \ref{Theorem3}.
Take Theorem \ref{Theorem1} for example,
in periods where active learning takes place,
all types of the sender mix between lying and telling the truth, with the highest-cost type telling the truth with strictly higher probability.
All other types lie with probability one if and only if
the receiver's belief attaches probability close to $1$ to the highest-cost type.
The highest-cost type's mixing during the active learning phase
allows each low-cost type  to \textit{rebuild} his reputation after milking it. It also reduces his reputational loss each time he extracts information rent,
which enables him to benefit from his persistent private information in the long run.
To provide incentives for all types of senders to mix in the active learning phase, I construct absorbing phases after which learning about the sender's type stops. The sender's action choices in the active learning phase affect
the time at which play reaches the absorbing phase, as well as his continuation payoff after learning stops.

\paragraph{Related Literature:}
My paper is related to the literature on
repeated communication games, strategic communication games with lying costs, and repeated games with incomplete information.

Repeated communication games are studied in the seminal works of
Sobel (1985) and Benabou and Laroque (1992).
Best and Quigley (2017) and Mathevet, Pearce and Stacchetti (2019)
use this framework to
rationalize the commitment assumption in Bayesian persuasion models.\footnote{Kuvalekar, Lipnowski and Ramos (2019) study repeated communication games in which receivers cannot observe the state realizations in the past. They show the equivalence between repeated communication games without feedback and one-shot communication games with capped money burning. Renault, Solan and Vieille (2013) and Margaria and Smolin (2018) focus on cases in which both players are patient. In those models, the sender can be punished by transferring payoffs to a patient receiver. This is not feasible in my model since receivers are myopic. Meng (2018) studies repeated communication games in which the receiver is \textit{patient} and \textit{can commit}. His result bounds the receiver's payoff from below when the sender has persistent private information about his preference.}

Mathevet, Pearce and Stacchetti (2019) adopt the commitment-type approach. They show that a patient sender can attain his Bayesian persuasion payoff if with positive probability, he is a \textit{commitment type} who mechanically communicates according to his optimal disclosure policy at every history. This differs from my baseline model in which all types of the sender are rational and strictly prefer to mislead receivers.
In terms of behavior, I show
that no type of rational sender communicates according to his optimal disclosure policy at every history.

Best and Quigley (2017) study repeated communication games \textit{without} persistent types. They propose a \textit{coin and cup
mechanism} that allows future receivers to perfectly monitor the sender's mixed actions, under which
the patient sender can attain his optimal commitment payoff.
In particular,
the sender has access to a
\textit{private randomization device}, the realization of which is \textit{not} observed by the current-period receiver, but is perfectly observed by the sender and will be truthfully disclosed to all future receivers.
However, the sender needs to have commitment power
since truthfully disclosing the realizations of this private randomization is against his own interest.
Compared to Best and Quigley (2017), my model focuses on situations in which
the sender
\textit{cannot commit} to act against his own interest, and consequently, the
receivers \textit{cannot perfectly monitor} the sender's mixed actions.

The fact that people face
psychological costs of lying has been established experimentally by
Gneezy (2005) and Gneezy, Kajackaite and Sobel (2018).
It has been incorporated in the strategic communication models of
Kartik, Ottaviani and Squintani (2007) and Kartik (2009).
Guo and Shmaya (2019) and Nguyen and Tan (2019)
study static communication games with lying costs in which
the sender receives private information according to a \textit{pre-committed information structure} before communicating with an uninformed receiver.
Their models nest Bayesian persuasion games of Kamenica and Gentzkow (2011), cheap talk games of Crawford and Sobel (1982), and communication games with information design in Ivanov (2010).\footnote{Another approach to bridge the gap between cheap talk games and Bayesian persuasion games is proposed by Lipnowski, Ravid and Shishkin (2019), who study static communication games in which the sender has transparent motives, and can commit to disclosure policies with positive probability.}
Their results provide a microfoundation for the sender's commitment to disclosure policies when his
lying cost is sufficiently large.

Those features contrast to my model in which the sender
automatically observes the state, in the sense that he
\textit{cannot} commit to receive coarser information.
My results suggest that non-ethical senders can attain their commitment payoffs \textit{only when} the highest lying cost belongs to some interval, not when the latter is large enough. Different from Guo and Shmaya (2019), the possibility of having a high cost of lying
may hurt a non-ethical sender through an outside option effect. In addition, the attainability of the sender's optimal commitment payoff relies on his ability to extract information rent in the long run, rather than relying entirely on the cost of lying.

My paper contributes to the study of repeated incomplete information games  pioneered by Aumann and Maschler (1995) and Hart (1985).
Shalev (1994) characterizes the set of equilibrium payoffs in private-value games with one-sided private information and no discounting.
P\c{e}ski (2014) extends Shalev's characterization to repeated games with discounting and allows for two-sided private information.
When the informed player is  patient and the uninformed player's discount factor is bounded away from one,
Cripps and Thomas (2003) show that
Shalev's result provides
a \textit{necessary condition}
for being an equilibrium payoff, but it is \textit{not sufficient} in general.

I provide  conditions that are \textit{both necessary and sufficient} for a patient sender's equilibrium payoff in a repeated communication game where receivers are impatient. My results are robust to perturbations of the receiver's discount factor. I also examine the common properties of the sender's behavior that uniformly apply across all sender-optimal equilibria. This raises novel questions provided that the existing literature on repeated games focuses mostly on equilibrium payoffs.

\section{Model}\label{sec2}
Time is discrete, indexed by $t =0,1,2,...$. A long-lived sender with discount factor $\delta \in (0,1)$ interacts with an infinite sequence of receivers, arriving one in each period and each plays the game only in the period she arrives.\footnote{My results are robust to perturbations of the receiver's discount factor. For example, when the sender communicates with one long-lived receiver whose discount factor
is strictly positive but close to zero.}

In period $t$, the realization of
$\omega_t \in \Omega \equiv \{h,l\}$ is privately observed by the sender.
The states $\{\omega_t\}_{t \in \mathbb{N}}$ are i.i.d., with
$p_h \in (0,1/2)$ the probability of $\omega_t=h$ and the value of $p_h$ being
commonly known.
The sender sends message $m_t \in M \equiv \{h,l\}$ to the period $t$ receiver. The latter takes an action $a_t \in A \equiv \{H,L\}$ after observing $m_t$ and the \textit{public history} $h^t \equiv \{a_s,\omega_s,m_s\}_{s=0}^{t-1} \in \mathcal{H}^t$.

The receiver's payoff is normalized to $0$ if $a_t=L$, her payoff is $1$ if $a_t=H$ and $\omega_t=h$, and her payoff is $-1$ if $a_t=H$ and $\omega_t=l$.
The sender's stage-game payoff is:
\begin{equation}\label{2.1}
    u_s(c,a_t,m_t,\omega_t)= \mathbf{1}\{a_t=h\} - c \mathbf{1}\{m_t \neq \omega_t\},
\end{equation}
where $c \in \mathcal{C}  \equiv \{c_1,c_2,...,c_n\} \subset [0,1)$ is interpreted as his \textit{psychological cost of lying}, and is incurred whenever the literal meaning of his message does not match the realized state.
Without loss of generality, I assume that $0 \leq c_n < c_{n-1} < ... c_2 <c_1 < 1$.
 In section \ref{sec5}, I relax the assumption that all types of the sender's lying cost are strictly less than $1$, and examine how the possibility of \textit{ethical-type senders} affects non-ethical-type sender's equilibrium payoffs.

I assume that $c$ is \textit{perfectly persistent} and is the sender's private information (or his \textit{type}).
The receivers entertain a full support prior belief $\pi \equiv (\pi_1,...,\pi_n)$,
with $\pi_j$ the probability of type $c_j$.
The distributions of
$c$ and $\{\omega_t\}_{t \in \mathbb{N}}$ are independent.
Under the assumption that $c_1<1$, all types of the sender are non-ethical in the sense that they have strict incentives to recommend the high action irrespective of the state
when facing a receiver who takes messages at their face values.

Let $\mathcal{H} \equiv \cup_{t=0}^{\infty} \mathcal{H}^t$ be the set of public histories.
The receiver's strategy is $\sigma_r : \mathcal{H} \times M \rightarrow \Delta (A)$.
Type $c$ sender's strategy is $\sigma_c: \mathcal{H} \times \Omega \rightarrow \Delta (M)$.
A strategy profile is
$\sigma \equiv \big((\sigma_c)_{c \in \mathcal{C}},\sigma_r\big)$, which
consists of a strategy for every type of the sender and a strategy for the receivers.

A \textit{Bayesian Nash Equilibrium} (BNE) is a strategy profile such that $\sigma_r$ maximizes each receiver's expected stage-game payoff at every $h^t$ that occurs with positive probability under $\sigma$, and $\sigma_c$ maximizes type $c$ sender's \textit{discounted average payoff}, given by:
\begin{equation}\label{2.2}
    \mathbb{E}^{(\sigma_c,\sigma_r)} \Big[
    \sum_{t=0}^{\infty} (1-\delta)\delta^t u_s(c,a_t,m_t,\omega_t)
    \Big],
\end{equation}
where $\mathbb{E}^{(\sigma_c,\sigma_r)}[\cdot]$ is the expectation under the probability measure induced by $(\sigma_c,\sigma_r)$.

A \textit{sequential equilibrium} consists of a strategy profile $\sigma\equiv \big((\sigma_c)_{c \in \mathcal{C}},\sigma_r\big)$ and an \textit{assessment}
$\boldsymbol{\mu}$
such that  (1) $\sigma_r$ maximizes the receiver's stage-game payoff at every history according to $\boldsymbol{\mu}$, (2)
for every $c \in\mathcal{C}$ and at every information set of type $c$ sender,
$\sigma_c$ maximizes his discounted average payoff in the continuation game
against $\sigma_r$,  and (3) there exists a sequence of completely mixed strategy profiles $\{\sigma^n\}_{n \in \mathbb{N}}$
and a sequence of assessments $\{\boldsymbol{\mu}^n\}_{n \in \mathbb{N}}$ such that
for every $n \in \mathbb{N}$,
$\boldsymbol{\mu}^n$ is derived from
$\sigma^n$
according to
Bayes Rule, and $(\sigma^n, \boldsymbol{\mu}^n) \rightarrow (\sigma,\boldsymbol{\mu})$ in the product topology.

\paragraph{Remark:}
The sender's stage-game payoff in (\ref{2.1}) embodies the \textit{non-consequentialism view} on lying, which postulates that
the sender incurs a cost of lying
whenever the literal meaning of his message fails to match the true state.
This includes situations in which his lie has caused no harm to the receivers, for example, when the receiver's action is independent of his message.
This view on lying is supported by the writings of Immanuel Kant and the recent work of Sobel (2020).\footnote{Kant
wrote in his influential article \textit{On the Supposed Right to Lie From Benevolent Motives} that \textit{To be truthful (honest) in all declarations is therefore a sacred unconditional command of reason, and not to be limited by any expediency.}
Sobel (2020) wrote that \textit{Lying depends on the existence of accepted meanings for messages, but does not require a model of how the audience
responds to messages}.}
It is adopted in the game theoretic models of Kartik, Ottaviani and Squintani (2007) and Kartik (2009).
In Appendix \ref{secD}, I extend my analysis to the \textit{consequentialism view} of lying, that
a lying cost is incurred \textit{only when} it has caused negative payoff consequences to the receiver (e.g., receiver takes the message at face value).\footnote{Martin Luther wrote \textit{...a lie out of necessity, a useful lie, a helpful lie, such lies would not be against God, he would accept them.} Gneezy (2005) shows experimentally that holding the sender's benefit from lying fixed, the propensity to lie decreases when the receiver becomes more credulous, or when the receiver's loss from the sender's lie increases.} My main takeaways remain valid under this alternative view.

\section{Results}\label{sec3}
\subsection{Highest Equilibrium Payoff}\label{sub3.1}
The sender's \textit{payoff} is an $n$-dimensional vector $v =(v_1,v_2,...,v_n)\in \mathbb{R}^n$, where $v_j$ is the discounted average payoff of type $c_j$. Let $v^* \equiv (v_1^*,...,v_n^*)$, with
\begin{equation}\label{3.1}
    v_j^* \equiv p_h \Big\{
    1+\frac{c_1-c_j}{2p_h +c_1 (1-2p_h)}
    \Big\}.
\end{equation}
Theorem \ref{Theorem1} characterizes every type of patient sender's \textit{highest equilibrium payoff}, and claims that the highest equilibrium payoffs for all types can be attained in the same equilibrium:
\begin{Theorem}\label{Theorem1}
There is no BNE such that type $c_1$ attains payoff strictly more than $v_1^*$. For every $\varepsilon>0$, there exists $\underline{\delta} \in (0,1)$ such that when $\delta >\underline{\delta}$:
\begin{enumerate}
\item There is no BNE such that type $c_j$ attains payoff more than $v_j^*+\varepsilon$ for some $j \in \{2,3,...,n\}$.
\item There exists a sequential equilibrium in which the sender attains payoff within $\varepsilon$ of $v^*$.
\end{enumerate}
\end{Theorem}
Theorem \ref{Theorem1} suggests that $v_j^*$ is type $c_j$ patient sender's \textit{highest equilibrium payoff} in the repeated communication game.
The two statements use different solution concepts to ensure that first, the payoff upper bounds in statement 1 apply under weak solution concepts, i.e., it applies to a broader set of outcomes. Second, the equilibria that approximately attain $v^*$
survive demanding refinements such as sequential equilibrium. Therefore, they are not driven by unreasonable beliefs off the equilibrium path.

The formula for type $c_j$'s highest equilibrium payoff has three implications.
First, every type of sender's highest equilibrium payoff depends \textit{only} on the receiver's prior belief about the i.i.d. state, his own cost of lying $c_j$, and the highest lying cost in the support of the receivers' prior belief $c_1$.
It \textit{does not} depend on the other types in the support of $\pi$ and the probability of each type.

Second, the type that has the highest lying cost
\textit{cannot} receive payoff that is strictly greater than $p_h$.
According to Fudenberg, Kreps and Maskin (1990) and Fudenberg and Levine (1994), $p_h$ is every type of sender's highest equilibrium payoff in a repeated game where his lying cost is common knowledge. Intuitively, this is because type $c_1$ is the \textit{most ethical type}, and as a result, he has no good candidate to imitate in the repeated game with persistent private information.

Third, $v_j^* >p_h$ for every $j \geq 2$. In another word, every type of sender
except for type $c_1$ \textit{strictly benefits} from persistent private information. This is somewhat puzzling since
a type of sender
obtaining discounted average payoff \textit{strictly greater} than $p_h$ requires this type to extract information rent, i.e., lying while receiving a receiver's trust. Each receiver's myopic incentive requires this type of sender to behave
differently from the other types, and as a result, extracting information rent reveals information about the sender's lying cost to future receivers and undermines the sender's informational advantage. As $\delta \rightarrow 1$, the number of periods in which the sender needs to extract information rent  (to attain a given discounted average payoff) grows without bound.
This raises a paradox that the sender can reveal information about his persistent type
 for unbounded number of periods while preserving his informational advantage in order to extract information rent in the future.
I explain how to construct equilibrium to resolve this conceptual puzzle in section \ref{sub4.3}, with technical details relegated to Appendix \ref{secA}.

Finally, when the highest lying cost $c_1$ converges to $1$,
$v_j^*$ converges to
\begin{equation}\label{3.2}
v_j^{**} \equiv p_h+p_h(1-c_j), \textrm{ for every } j \in \{1,2,...,n\},
\end{equation}
where $v_j^{**}$ is type $c_j$ sender's payoff in a static \textit{Bayesian persuasion} game taking his psychological cost of lying into account. Intuitively, when the sender can commit to communication rule $\boldsymbol{\alpha}:\Omega \rightarrow \Delta (M)$,
his optimal commitment requires him to send message $h$ with probability $1$ when the state is $h$, and to send message $h$
with probability $\frac{p_h}{1-p_h}$ when the state is $l$.
This disclosure policy
makes the receiver indifferent between actions $H$ and $L$ upon receiving message $h$.


The above implication of
Theorem \ref{Theorem1} provides a microfoundation for the sender's commitment to information structures in Bayesian persuasion models.
In environments where all types of the sender
are rational and have strict incentives to mislead receivers,
as long as there exists one type of the sender whose lying cost is close to his benefit from the receiver's high action,
all types can (approximately) attain their optimal commitment payoffs in a repeated communication game.

Compared to commitment-type models which assume that with positive probability, the sender is a \textit{commitment type} that adopts his optimal information disclosure policy in every period,
my approach addresses
critiques on commitment-type senders who mechanically communicate according to some stochastic disclosure policies in every period, given that
whether
such commitment behaviors can arise from maximizing reasonable payoff functions is somewhat questionable.

Theorem \ref{Theorem1} also characterizes
the extent to which
the sender's commitment power
can be \textit{partially restored} via persistent private information when he repeatedly communicates with multiple receivers over time.
My result provides a tractable formula for every type's highest equilibrium payoff, which is between his highest payoff in the repeated complete information game and his payoff in the one-shot Bayesian persuasion game.
My formula depends only on primitives that have clear economic interpretations and clarifies the role of persistent private information. It bridges the gap between existing results on repeated complete information games, Bayesian persuasion games, and reputation games with commitment types.

\subsection{Equilibrium Behavior}\label{sub3.2}
My next result clarifies the distinction between a strategic-type sender who faces high lying cost and a commitment-type who
mechanically uses his optimal disclosure policy at every history.

Let $\mathbf{a}: \Omega \rightarrow M$ be a \textit{pure stage-game action} for a given type of sender, and let $\mathbf{b}: M \rightarrow A$ be a \textit{pure stage-game action} for a receiver. Let $\mathbf{A}$ and $\mathbf{B}$ be the sets of $\mathbf{a}$ and $\mathbf{b}$, respectively.
Let
$\mathbf{a}^H$ be the \textit{honest strategy} and let $\mathbf{a}^L$
be the \textit{lying strategy} for the sender:
\begin{equation}\label{sender}
\mathbf{a}^H (\omega) \equiv \left\{ \begin{array}{ll}
h & \textrm{ if } \omega=h\\
l & \textrm{ if } \omega=l,
\end{array} \right.
\quad
\mathbf{a}^L (\omega) \equiv \left\{ \begin{array}{ll}
h & \textrm{ if } \omega=h\\
h & \textrm{ if } \omega=l.
\end{array} \right.
\end{equation}
Let $\mathbf{b}^T$ be the \textit{trusting strategy} and let $\mathbf{b}^N$
be the \textit{non-trusting strategy} for the receiver:
\begin{equation}\label{receiver}
\mathbf{b}^T (m) \equiv \left\{ \begin{array}{ll}
H & \textrm{ if } m=h\\
L & \textrm{ if } m=l,
\end{array} \right.
\quad \mathbf{b}^N (m) \equiv \left\{ \begin{array}{ll}
L & \textrm{ if } m=h\\
L & \textrm{ if } m=l.
\end{array} \right.
\end{equation}

Abusing notation, let $\sigma_{c_j}: \mathcal{H} \rightarrow \Delta (\mathbf{A})$ be type $c_j$ sender's strategy,
and let $\sigma_{r}: \mathcal{H} \rightarrow \Delta (\mathbf{B})$ be
the receiver's strategy.
My next result shows that in every \textit{sender-optimal equilibrium},
no type of the sender mixes between
$\mathbf{a}^H$ and $\mathbf{a}^L$ at all on-path histories.
\begin{Theorem}\label{Theorem2}
Suppose $n \geq 2$. For every small enough $\varepsilon>0$, there exists $\underline{\delta} \in (0,1)$ such that when $\delta >\underline{\delta}$, for every BNE in which the sender attains payoff greater than $(v_1^*-\varepsilon,...,v_n^*-\varepsilon)$,
no type of the sender plays both $\mathbf{a}^H$ and $\mathbf{a}^L$ with positive probability at all on-path histories.
\end{Theorem}
\begin{proof}[Proof of Theorem 2:] Suppose toward a contradiction that there exists a type $c_j$ that
plays both $\mathbf{a}^H$ and $\mathbf{a}^L$ with positive probability at every on-path history.
Then playing $\mathbf{a}^H$ at every on-path history and playing $\mathbf{a}^L$ at every on-path history are both his best replies against $\sigma_r$. This further implies that for every $i<j$, type $c_i$ plays $\mathbf{a}^H$ with probability $1$ at every on-path history.\footnote{Different from the binary action game in Pei (2019), it is \textit{not} true that for every $k<j$, type $c_k$ plays $\mathbf{a}^L$ with probability $1$ at every on-path history. This is because the sender has other stage-game actions, such as lying in both states, from which case he suffers strictly higher lying cost compared to playing $\mathbf{a}^L$.}
The above statement uses the observation
that the sender's benefit from changing his stage-game action is a separable function with respect to his type and the receiver's stage-game action.

I consider two cases separately.
First, if $j \geq 2$, then type $c_1$ plays $\mathbf{a}^H$ with probability $1$ at every on-path history. Therefore, type $c_2$ separates from type $c_1$ the first time he sends message $h$ in state $l$, after which
he becomes the highest-cost type in the support of the receivers' posterior belief.
According to Proposition \ref{Prop1}, his continuation value is no more than $p_h$.
As a result, type $c_2$'s expected payoff in period $0$ is no more than $(1-\delta)+\delta p_h$, which
is strictly lower than $v_2^*$ as $\delta \rightarrow 1$. This contradicts the presumption that type $c_2$'s equilibrium payoff exceeds $v_2^*-\varepsilon$.

Second, if $j=1$, then type $c_1$ finds it optimal to play $\mathbf{a}^L$ in every period. Since the sender's equilibrium payoff is within $\varepsilon$ of $v^*$, type $c_1$'s payoff is at least $v_1^*-\varepsilon$ by playing $\mathbf{a}^L$ in every period, and type $c_2$'s payoff from doing so is no more than $v_2^*+\varepsilon$.
Since $p_h<1/2$, the receiver's stage-game action of playing $H$ following every message is strictly suboptimal, and cannot be played at any on-path history. Among the remaining three receiver stage-game actions, the sender's stage-game payoff is $1-(1-p_h)c$ under $(\mathbf{a}^L,\mathbf{b}^T)$, is $-(1-p_h)c$
under $(\mathbf{a}^L,\mathbf{b}^N)$ and $(\mathbf{a}^L,\mathbf{b}^O)$, where
\begin{displaymath}
\mathbf{b}^O (m) \equiv \left\{ \begin{array}{ll}
H & \textrm{ if } m=l\\
L & \textrm{ if } m=h.
\end{array} \right.
\end{displaymath}
Let $Q_L$ be the occupation measure of $(\mathbf{a}^L,\mathbf{b}^T)$ when the sender plays $\mathbf{a}^L$ in every period and the receiver plays according to $\sigma_r$. Type $c_1$'s equilibrium payoff is:
\begin{equation*}
    Q_L \Big( 1-(1-p_h)c_1 \Big) -(1-Q_L)(1-p_h)c_1.
\end{equation*}
The presumption that type $c_1$'s equilibrium payoff is more than $p_h-\varepsilon$ yields a lower bound on $Q_L$, which is:
\begin{equation}\label{QL}
    Q_L \geq p_h +(1-p_h) c_1-\varepsilon.
\end{equation}
Type $c_2$'s payoff
by playing $\mathbf{a}^L$ in every period is
    $Q_L \big( 1-(1-p_h)c_2 \big) -(1-Q_L)(1-p_h)c_2$.
Plugging in (\ref{QL}), one can obtain that type $c_2$'s equilibrium payoff is at least:
\begin{equation}\label{QL2}
    p_h+(1-p_h) (c_1-c_2) - \varepsilon
\end{equation}
Given that $p_h<1/2$, the lower bound on type $c_2$'s equilibrium payoff (\ref{QL2}) is strictly greater than $v_2^*+\varepsilon$ when $\varepsilon$ is small enough.
This contradicts statement 1 of Theorem \ref{Theorem1} that type $c_2$'s equilibrium payoff when $\delta$ is large enough cannot exceed $v_2^*+\varepsilon$.
\end{proof}
Since the sender's optimal disclosure policy in a static Bayesian persuasion game is to play $\mathbf{a}^L$ with probability $\rho^* \equiv \frac{p_h}{1-p_h}$, and $\mathbf{a}^H$ with probability $1-\rho^*$, a direct implication of Theorem \ref{Theorem2} is that no matter how large the sender's lying cost is,
there exists no BNE in which he communicates according to his optimal disclosure policy at all on-path histories.

Formally, I call $\boldsymbol{\alpha} \in \Delta (\mathbf{A})$ an $\varepsilon$-optimal disclosure policy if $\mathbf{b}^T$ is a \textit{strict best reply} against $\boldsymbol{\alpha}$, and $\boldsymbol{\alpha}$ belongs to an $\varepsilon$-neighborhood of $\rho^* \mathbf{a}^L +(1-\rho^*) \mathbf{a}^H$.\footnote{$\rho^* \mathbf{a}^L +(1-\rho^*) \mathbf{a}^H$ is \textit{not} an $\varepsilon$-optimal disclosure policy since it violates the first requirement. This is because according to Fudenberg and Levine (1992) and Gossner (2011), a long-run player's guaranteed equilibrium payoff is his payoff from playing his commitment action and his opponents play the best response that \textit{minimizes} his payoff.}
\begin{Corollary}
For every small enough $\varepsilon>0$, there exists $\underline{\delta} \in (0,1)$ such that when $\delta >\underline{\delta}$, in every BNE, no type of the sender plays  an $\varepsilon$-optimal disclosure policy at every on-path history.
\end{Corollary}
Corollary 1 applies to any number of types.
 Its conclusion contrasts to the commitment-type sender in Mathevet, Pearce and Stacchetti (2019) who mechanically communicates according to his optimal disclosure policy at every on-path history.
\begin{proof}[Proof of Corollary 1:] First, consider the case in which $n=1$. If the sender uses an
$\varepsilon$-optimal disclosure policy at every on-path history, then by definition, the receiver has a strict incentive to play $\mathbf{b}^T$ at every on-path history. By playing $\mathbf{a}^L$ in every period, the sender obtains discounted average payoff $p+(1-p)c$, which is strictly greater than $p$ and leads to a contradiction.

Next, consider the case in which $n \geq 2$. Suppose toward a contradiction that there exists a type $c_k  \in C$ who plays an $\varepsilon$-optimal disclosure policy at every on-path history in equilibrium. Given that the public history
can statistically identify the sender's stage-game action $\mathbf{a} \in \mathbf{A}$,
the results in Fudenberg and Levine (1992) and Gossner (2011)
imply that for every $j \in \{1,2,...,n\}$, type $c_j$ can guarantee payoff approximately $v_j^{**}$ by playing the equilibrium strategy of type $c_k$. By definition, $v_j^{**}>v_j^*-\varepsilon$ for every $j$.
This contradicts the conclusion of Theorem \ref{Theorem2}, since every $\varepsilon$-optimal disclosure policy attaches positive probabilities to $\mathbf{a}^H$ and $\mathbf{a}^L$, and the sender's equilibrium payoff strictly exceeds $(v_1^*-\varepsilon,...,v_n^*-\varepsilon)$ when $\delta$ is large enough.
\end{proof}
 My proofs of Theorem \ref{Theorem2} and Corollary 1 suggest that
both conclusions extend to a type of sender whose cost of lying exactly offsets his benefit from the receiver's high action, i.e., $c_1=1$.
Intuitively, this is because a strategic-type sender who has lying cost $1$
strictly benefits from the receiver's high action, which differs from
a commitment type who does not care about payoffs. As a result,
the strategic-type sender's indifference between $\mathbf{a}^H$ and $\mathbf{a}^L$ at a given history $h^t$ introduces constraints on the receiver's strategies in the continuation game. This in turn leads to constraints on other types of senders' incentives and payoffs.
\section{Proof of Theorem 1}\label{sec4}
In section 4.1, I construct a constrained optimization problem with $v_j^*$ equals its optimal value. Then I map the outcomes of the repeated communication game to this constrained optimization problem, according to which the objective function coincides with type $c_j$ sender's discounted average payoff. In sections 4.2 and 4.3, I show that the constraints in the optimization problem are necessary for any equilibrium outcome of the repeated game. In section 4.4, I construct sequential equilibrium in which the patient sender approximately attains payoff $v^*$.
\subsection{Payoff Upper Bound: $v_j^*$ as a Constrained Optimization Problem}
Recall the definitions of pure stage-game actions for the sender and the receiver.
Let $u_s(c,\mathbf{a},\mathbf{b})$ and $u_r(\mathbf{a},\mathbf{b})$ be the sender's and the receiver's stage-game payoff functions, respectively. Both are naturally extended to mixed actions. Lemma \ref{L4.1} relates $v_j^*$ to a constrained optimization problem defined via the stage game:
\begin{Lemma}\label{L4.1}
For given $j \in \{1,2,...,n\}$, the value of the following constrained optimization problem is $v_j^*$:
\begin{equation}\label{3.3}
    \max_{\gamma \in \Delta (\mathbf{A} \times \mathbf{B})} \sum_{(\mathbf{a},\mathbf{b}) \in \mathbf{A} \times \mathbf{B}} \gamma (\mathbf{a},\mathbf{b}) u_s(c_j,\mathbf{a},\mathbf{b}),
\end{equation}
subject to:
\begin{equation}\label{3.4}
 \sum_{(\mathbf{a},\mathbf{b}) \in \mathbf{A} \times \mathbf{B}} \gamma (\mathbf{a},\mathbf{b}) u_s(c_1,\mathbf{a},\mathbf{b}) \leq p_h,
\end{equation}
and for every $\mathbf{b} \in \mathbf{B}$ that
the marginal distribution of
$\gamma$ on $\mathbf{B}$ attaches positive probability to,
\begin{equation}\label{3.5}
 \mathbf{b} \in \arg\max_{\mathbf{b'} \in \mathbf{B}}   u_r(\gamma(\cdot|\mathbf{b}),\mathbf{b'}),
\end{equation}
where $\gamma(\cdot|\mathbf{b}) \in \Delta (\mathbf{A})$ is the distribution
\textit{conditional on} the receiver's stage-game action being $\mathbf{b}$.
\end{Lemma}
Let
\begin{equation}\label{rho}
\rho^* \equiv \frac{p_h}{1-p_h},
\end{equation}
and recall the definitions of $\mathbf{a}^H$, $\mathbf{a}^L$, $\mathbf{b}^T$, and $\mathbf{b}^N$ in
(\ref{sender}) and (\ref{receiver}), respectively.
The following distribution over stage-game action profiles, denoted by $\gamma^*$, attains the optimal value, with:
\begin{equation*}
    \gamma^*(\mathbf{a}^L,\mathbf{b}^N)=\frac{\rho^*(1-c_1)}{\rho^* (1-c_1)+c_1},
\end{equation*}
\begin{equation*}
    \gamma^* (\mathbf{a}^L,\mathbf{b}^T)=\frac{\rho^*c_1}{\rho^* (1-c_1)+c_1},
\end{equation*}
and
\begin{equation*}
    \gamma^* (\mathbf{a}^H,\mathbf{b}^T)=\frac{(1-\rho^*)c_1}{\rho^* (1-c_1)+c_1}.
\end{equation*}
Players' stage-game payoffs under these stage-game action profiles are given by:
\begin{center}
\begin{tabular}{| c | c | c |}
  \hline
  $-$ & $\mathbf{b}^T$ & $\mathbf{b}^N$ \\
  \hline
  $\mathbf{a}^H$ & $p_h,p_h$ & $0,0$ \\
  \hline
  $\mathbf{a}^L$ & $p_h+(1-c)(1-p_h),2p_h-1$ & $-c(1-p_h),0$ \\
  \hline
\end{tabular}
\end{center}
and a graphical illustration of the two constraints as well as $v_j^*$ can be found in Figure 1.
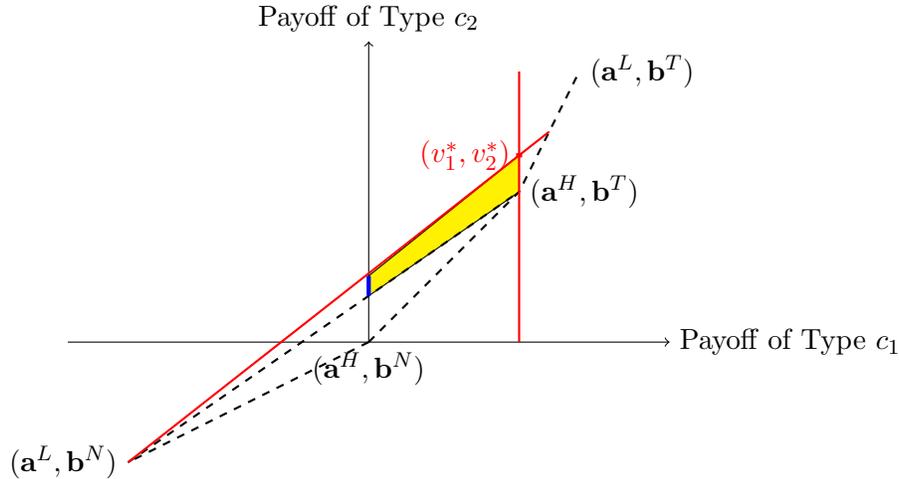
\begin{figure}\label{Figure1}
\begin{center}
\begin{tikzpicture}[scale=0.4]
\draw [fill=yellow] (0,1.53846)--(0,2.2)--(5,6.214)--(5,5)--(0,1.53846);
\draw [->] (-10,0)--(10,0)node[right]{Payoff of Type $c_1$};
\draw [->] (0,0)--(0,10)node[above]{Payoff of Type $c_2$};
\draw [dashed, thick] (0,0)--(5,5)node[right]{$(\mathbf{a}^H,\mathbf{b}^T)$}--(7,9)node[right]{$(\mathbf{a}^L,\mathbf{b}^T)$};
\draw [dashed, thick] (-8,-4)--(5,5);
\draw [dashed, thick] (0,0)node[below]{$(\mathbf{a}^H,\mathbf{b}^N)$}--(-8,-4)node[left]{$(\mathbf{a}^L,\mathbf{b}^N)$};
\draw [red, thick] (-8,-4)--(6,7);
\draw [red, thick] (5,0)--(5,9);
\draw [ultra thick, red] (4.9,6.214)--(5.1,6.214)node[left]{$(v_1^*,v_2^*)$};
\draw [ultra thick, blue] (0,1.53846)--(0,2.2);
\end{tikzpicture}
\caption{The sender's highest equilibrium payoff vector $(v_1^*,v_2^*)$ when there are two types. The two red lines capture the two constraints, and the intersection between them is $v^*$.}
\end{center}
\end{figure}

Next, I relate this constrained optimization problem to outcomes in the repeated communication game.
Recall that
$\sigma_c : \mathcal{H} \rightarrow \Delta (\mathbf{A})$ is type $c$ sender's strategy, and
$\sigma_r : \mathcal{H} \rightarrow \Delta (\mathbf{B})$ is the receiver's strategy.
For any given \textit{strategy profile} $\sigma \equiv \big((\sigma_{c})_{c \in \mathcal{C}}, \sigma_r \big)$,
let
\begin{equation}\label{3.6}
    \gamma^j(\mathbf{a},\mathbf{b}) \equiv  \mathbb{E}^{(\sigma_{c_j},\sigma_r)} \Big[
    \sum_{t=0}^{\infty} (1-\delta)\delta^t \mathbf{1} \{
(\mathbf{a}_t,\mathbf{b}_t) =   (\mathbf{a},\mathbf{b})
     \}
    \Big], \textrm{ for every } (\mathbf{a},\mathbf{b}) \in \mathbf{A} \times \mathbf{B}.
\end{equation}
This defines a distribution over stage-game action profiles, denoted by $\gamma^j \in \Delta (\mathbf{A} \times \mathbf{B})$.
By construction, type $c_j$ sender's discounted average payoff in the repeated game under strategy profile $\sigma$ equals his expected stage-game payoff under distribution $\gamma^j$, which is the objective function (\ref{3.3}) once replacing $\gamma$ with $\gamma^j$.
Therefore, as long as $\gamma^j$ satisfies constraints (\ref{3.4}) and (\ref{3.5}),
type $c_j$ sender's discounted average payoff in the repeated game under strategy profile $\sigma$ cannot exceed $v_j^*$.

\subsection{Necessity of Constraint (\ref{3.4})}\label{sub4.1}
The necessity of constraint (\ref{3.4}) is implied by type $c_1$'s equilibrium payoff being no more than $p_h$.
This is because
the left-hand-side of (\ref{3.4}) is type $c_1$'s payoff by deviating type $c_j$'s equilibrium strategy $\sigma_{c_j}$. The latter cannot exceed $p_h$ if type $c_1$'s equilibrium payoff is at most $p_h$.
\begin{Proposition}\label{Prop1}
For every Bayesian Nash Equilibrium $\sigma$ and for every $h^t$ that occurs with positive probability under $\sigma$, if $c_i$ is the highest-cost type in the support of the receiver's belief at $h^t$, then type $c_i$'s continuation payoff at $h^t$ is at most $p_h$.
\end{Proposition}
A caveat is that even when
there is \textit{only one type} in the support of receivers' belief, the conclusion of Proposition \ref{Prop1} \textit{does not} follow from the folk theorem results in Fudenberg, Kreps and Maskin (1990) and Fudenberg and Levine (1994). This is because a type that occurs with zero probability at a given history is \textit{not} equivalent to a type that is excluded from the type space.
In particular, zero probability types
may occur with strictly positive probability at some future off-path histories. This will in turn affect players' incentives and payoffs at on-path histories.\footnote{Osborne and
Rubinstein (1990) present an example in which types that are included
in the type space but occur with zero prior probability
are different from types that are excluded from the type space. They introduce a refinement called ``\textit{never dissuaded once convinced}'' to rule out such differences in outcomes. Madrigal, Tan and Werlang (1987) construct a finite extensive form game in which there exists no equilibrium that satisfies the requirement ``\textit{the support of beliefs at an information set be contained in the supports of beliefs at preceding information sets}''.}
\begin{proof}[Proof of Proposition 1:]
Let $C(h^t)$ be the support of the receiver's posterior belief after observing $h^t$ but before observing $m_t$. I show Proposition \ref{Prop1} by induction on $|C(h^t)|$, namely, the number of types in the support of the receiver's belief. My proof consists of two steps.
\paragraph{Step 1:} Suppose $|C(h^t)|=1$, then
$C(h^s)=C(h^t)$
for every $h^s$ such that $h^s \succ h^t$ and $h^s$ occurs with positive probability under $\sigma$.
Let $c_i$ be the only type in $C(h^t)$.
Strategy $\widetilde{\sigma}_{c_i}: \mathcal{H} \times \Omega \rightarrow \Delta (M)$ defined below also best replies against the receiver's equilibrium strategy $\sigma_r$:
\begin{equation}\label{3.7}
    \widetilde{\sigma}_{c_i} (h^s)(\omega) \equiv \left\{ \begin{array}{ll}
m & \textrm{ if }  \sigma_{c_i}(h^s)(\omega) \textrm{ attaches positive probability to } m \textrm{ and } \sigma_r (h^s)(m)=L\\
\sigma_{c_i}(h^t)(\omega) & \textrm{ otherwise}.
\end{array} \right.
\end{equation}
If both messages induce action $L$ for sure and both messages are sent with positive probability by $\sigma_{c_i}$ at an information set, then
pick any message for the sender.

By definition, type $c_i$'s payoff under $(\widetilde{\sigma}_{c_i}, \sigma_r)$ equals his continuation payoff at $h^t$.
If type $\theta_i$ plays according to $\widetilde{\sigma}_{\theta_i}$ against $\sigma_2$,
then his expected payoff at every on-path history following $h^t$ cannot
exceed $p_h$.
Therefore, type $c_i$'s  discounted average payoff at $h^t$ cannot exceed $p_h$.

\paragraph{Step 2:} I show that if the conclusion holds for histories where $|C(h^t)|\leq k$, then it also holds for histories where $|C(h^t)|=k+1$.

Let $\mathcal{H}^{\sigma}$ be the set of public histories that occur with positive probability under $\sigma$.
Let $c_i$ be the type that has the highest lying cost in $C(h^t)$.
Given type $c_i$'s equilibrium strategy $\sigma_{c_i}$,
and recall the definition of $\widetilde{\sigma}_{c_i}$
in (\ref{3.7}), which is one of type $c_i$'s best replies against the receiver's equilibrium strategy $\sigma_r$.
Let $\mathcal{H}^{(\widetilde{\sigma}_{c_i},\sigma_r)}$ be the set of histories that occur with positive probability under
$(\widetilde{\sigma}_{c_i},\sigma_r)$, which I partition into two subsets:
\begin{enumerate}
  \item Outcome $(\omega=l,a=H)$ has never occurred before.
  \item Outcome $(\omega=l,a=H)$ has occurred before.
\end{enumerate}
Suppose type $c_i$ sender plays according to $\widetilde{\sigma}_{c_i}$,
at every history $h^t$ that belongs to the first subset, he has never received positive stage-game payoff when $\omega_s=l$.
At every history $h^t$ that belongs to the second subset, but its immediate predecessor $h^{t-1}$ belongs to the first subset,
according to the definition of $\widetilde{\sigma}_{c_i}$,
\begin{itemize}
  \item there exists a message that induces action $H$ with positive probability at $h^{t-1}$,
  \item type $c_i$ sends that message at $h^{t-1}$ with probability $1$.
\end{itemize}
The receiver's incentive to play $H$ after receiving the aforementioned message at $h^{t-1}$
implies the existence of type $c_j \in C(h^{t-1})$ with $c_j \neq c_i$ that sends the other message, denoted by $m'$, with positive probability when $\omega_{t-1}=l$ at $h^{t-1}$. Since $c_i$ is the type with the highest lying cost in set $C(h^{t-1})$, we have $c_j < c_i$. Moreover, at history
$(h^{t-1},(l,m',L))$, type $\theta_i$ occurs with probability $0$, so $|C(h^{t-1},(l,m',L))| \leq |C(h^{t-1})|-1$.
According to the induction hypothesis, there exists $c_{\tau} \in C(h^{t-1},(l,m',L))$ such that type $c_{\tau}$'s continuation payoff at $(h^{t-1},(l,m',L))$ is no more than $p_h$. Given type $c_{\tau}$'s incentive to send message $l$ at $h^{t-1}$ when $\omega_t=l$, we have:
\begin{equation*}
\delta p_h \geq \delta V_{c_{\tau}}(h^{t-1},(l,m',L))
\geq (1-\delta) \big(\Pr(a_{t-1}=H|h^{t-1},m_{t-1}\neq m') - c_{\tau} \big)
\end{equation*}
\begin{equation*}
+\delta \Pr(a_{t-1}=H|h^{t-1},m_{t-1}\neq m') V_{c_{\tau}}(h^{t-1},(l,h,H))
+\delta \Pr(a_{t-1}=L|h^{t-1},m_{t-1}\neq m') V_{c_{\tau}}(h^{t-1},(l,h,L)).
\end{equation*}
Since type $c_{\tau}$'s stage-game payoff is no less than type $c_i$'s stage-game payoff at every history, type $c_i$'s continuation payoff by not sending message $m'$
when $\omega_{t-1}=l$ at history $h^{t-1}$ is no more than $\delta p_h$. Combining the conclusions at the two subsets of histories, we know that type $c_i$'s continuation payoff is no more than $p_h$ at every $h^t$ with $|C(h^t)| \leq k+1$.

Since the number of types is finite, the above induction argument implies that at every on-path history $h^t$, the highest-cost type in the support of the receivers' posterior belief at $h^t$ receives continuation payoff no more than $p_h$.
\end{proof}

\subsection{Necessity of Constraint (\ref{3.5})}\label{sub4.2}
I establish the necessity of (\ref{3.5}) in two steps.
Proposition \ref{Prop2} shows that $\gamma^j$
satisfies an $\varepsilon$-relaxed version of constraint (\ref{3.5})
when $\delta$ is above some cutoff. Proposition \ref{Prop3} shows that the value of the $\varepsilon$-constrained program converges to the value of the original program as $\varepsilon$ vanishes to $0$.
\begin{Proposition}\label{Prop2}
For every $\varepsilon>0$, there exists $\underline{\delta} \in (0,1)$ such that when $\delta > \underline{\delta}$,
for every Bayesian Nash Equilibrium and every $j \in \{1,2,...,n\}$, if $\gamma^j$ attaches probability more than $\varepsilon$ to
$\mathbf{b}$, then $\mathbf{b}$ is an $\varepsilon$-best reply against $\gamma^j(\cdot|\mathbf{b})$.
\end{Proposition}
\begin{proof}[Proof of Proposition 2:]
For every $h^{\tau} \in \mathcal{H}^{\sigma}$, let $\sigma_{c_j}(h^{\tau}) \in \Delta (\mathbf{A})$ be the distribution over stage-game pure actions prescribed by $\sigma_{c_j}$ at $h^{\tau}$, and let $\alpha(h^{\tau}) \in \Delta (\mathbf{A})$ be the receiver's belief about the sender's stage-game pure actions at $h^{\tau}$.
Since future receivers can perfectly observe the past state realizations and the sender's messages, the public signals can \textit{statistically identify} the sender's pure stage-game action $\mathbf{a} \in \mathbf{A}$,
Gossner (2011)'s result provides an upper bound on the expected sum of the receivers' \textit{one-step ahead prediction errors}, measured by the Kullback-Leibler divergence (KL divergence for short) between
$\sigma_{c_j}(h^{\tau})$ and $\alpha(h^{\tau})$:
\begin{equation}\label{3.8}
  \mathbb{E}^{(\sigma_{c_j},\sigma_r)} \Big[   \sum_{\tau=0}^{+\infty} d(\sigma_{c_j}(h^{\tau})||\alpha(h^{\tau}))
    \Big]
    \leq -\lambda \log \pi_j,
\end{equation}
where $d(\cdot||\cdot)$ is the KL-divergence,
$\pi_j$ is the prior probability of type $c_j$, and
$\lambda>0$ is a constant.

Inequality (\ref{3.8}) implies that for every $\xi >0 $, the expected number of periods such that $ d(\sigma_{c_j}(h^{\tau})||\alpha(h^{\tau})) > \xi$ is no more than
\begin{equation}\label{3.9}
    T(\xi) \equiv \Big\lceil \lambda \frac{-\log \pi_j}{\xi}\Big\rceil.
\end{equation}
Let $\sigma_r(h^{\tau}) \in \Delta (\mathbf{B})$ be the distribution over receiver's pure stage-game action prescribed by $\sigma_r$ at $h^{\tau}$.
Let $\mathbf{B}^{\sigma_r}(h^{\tau})$ be the support of $\sigma_r(h^{\tau})$.
Since the receiver plays a stage-game best reply against her expectation over the sender's stage-game action, we have:
\begin{eqnarray}\label{3.10}
& & \mathbb{E}^{(\sigma_{c_j},\sigma_r)} \Big[
    \sum_{\tau=0}^{\infty}
    (1-\delta)\delta^{\tau} \mathbf{1} \Big\{
  \mathbf{b} \in   \mathbf{B}^{\sigma_r}(h^{\tau}) \textrm{ but } \mathbf{b} \textrm{ does not best reply to any } \alpha \textrm{ with }
     ||\alpha-\sigma_{c_j}(h^{\tau})||
     \leq \sqrt{2\xi}
    \Big\}
    \Big]
     {}
\nonumber\\
&\leq& {}
    \mathbb{E}^{(\sigma_{c_j},\sigma_r)} \Big[
    \sum_{\tau=0}^{\infty}
    (1-\delta)\delta^{\tau} \mathbf{1} \Big\{
    \mathbf{b} \in   \mathbf{B}^{\sigma_r}(h^{\tau}) \textrm{ but } \mathbf{b} \textrm{ does not best reply to any } \alpha \textrm{ with } d(\sigma_{c_j}(h^{\tau})||\alpha) \leq \xi
    \Big\}
    \Big]
    {}
\nonumber\\
&\leq& {}
    \mathbb{E}^{(\sigma_{c_j},\sigma_r)} \Big[
    \sum_{\tau=0}^{\infty}
    (1-\delta)\delta^{\tau} \mathbf{1} \Big\{
   d(\sigma_{c_j}(h^{\tau})||\alpha(h^{\tau})) > \xi
    \Big\}
    \Big]
\leq
    1-\delta^{ T(\xi)}.
\end{eqnarray}
The first inequality comes from the Pinsker's inequality, the second inequality holds since $\mathbf{b}$ best replies against
$\alpha(h^{\tau})$, and the third inequality comes from (\ref{3.8}) and (\ref{3.9}).

Recall the definition of $\gamma^j \in \Delta (\mathbf{A} \times \mathbf{B})$.
Let $\beta^j$ be the marginal distribution of $\gamma^j$ on $\mathbf{B}$, and let $\gamma^j(\cdot|\mathbf{b})$ be the distribution over
$\mathbf{A}$ conditional on $\mathbf{b}$ under joint distribution $\gamma^j$.
Let $\mathcal{A} (\mathbf{b}) \subset \Delta (\mathbf{A})$ be the set of sender's mixed stage-game actions that $\mathbf{b}$ best replies against.
Consider any $\mathbf{b} \in \mathbf{B}$ with the property that the Hausdorff distance between
$\gamma^j(\cdot|\mathbf{b})$
and set $\mathcal{A} (\mathbf{b})$
is more than $\varepsilon$. I denote this distance by $D$.
For every $\eta>0$, let
$\mathcal{A}^{\eta} (\mathbf{b})$ be the set of elements in $\Delta (\mathbf{A})$ whose Hausdorff distance to $\mathcal{A}(\mathbf{b})$ is no more than $\eta$.
Since the Hausdorff distance between any two points in $\Delta (\mathbf{A})$ is at most $1$, for any distribution over the sender's mixed stage-game actions $\rho \in \Delta (\Delta (\mathbf{A}))$ that has countable support $\{\alpha^i\}_{i \in \mathbb{N}}$, and satisfies:
    $\sum_{i \in \mathbb{N}} \rho(\alpha_1^i) \alpha^i =\alpha^j(\cdot|\mathbf{b})$,
we have:
\begin{equation}\label{3.11}
    \sum_{\alpha^i \notin \mathcal{A}^{\eta} (\mathbf{b}) } \rho(\alpha^i) \geq \frac{D-\eta}{1+D-\eta}.
\end{equation}
Therefore:
\begin{eqnarray}\label{3.12}
& &  \mathbb{E}^{(\sigma_{c_j},\sigma_r)} \Big[
    \sum_{\tau=0}^{\infty}
    (1-\delta)\delta^{\tau} \mathbf{1} \Big\{
     \mathbf{b} \in   \mathbf{B}^{\sigma_r}(h^{\tau}) \textrm{ but } \mathbf{b} \textrm{ does not best reply to any } \alpha \textrm{ with }
     ||\alpha-\sigma_{c_j}(h^{\tau})||
     \leq \eta
    \Big\}
    \Big]      {}
\nonumber\\
&\geq& {} \frac{\beta^j(\mathbf{b})(D-\eta)}{1+D-\eta}.
\end{eqnarray}
Pick $\eta \equiv\frac{D}{2}$ and $\xi \equiv \frac{D^2}{8}$, we have $\sqrt{2 \xi}=\frac{D}{2}$. Therefore,
(\ref{3.10}) and (\ref{3.12}) together imply that
for every strategy profile that is an equilibrium under
 discount factor $\delta$, we have:
\begin{equation}\label{3.13}
  \beta^j(\mathbf{b}) \leq  \Big(1-\delta^{T(\frac{D^2}{8})} \Big) \frac{1+D/2}{D/2}.
\end{equation}
Since $D \geq \varepsilon$, there exists $\underline{\delta} \in (0,1)$ such that
the RHS of (\ref{3.13}) is less than
$\varepsilon$
for every $\delta \in (\underline{\delta},1)$. That is to say, for every $\mathbf{b} \in \mathbf{B}$ such that
$\mathbf{b}$ is not an $\varepsilon$-best reply against
$\gamma^j(\cdot|\mathbf{b})$, the marginal distribution $\beta^j$ attaches probability less than $\varepsilon$ to $\mathbf{b}$.
\end{proof}
Let $v_j^{\varepsilon}$ be the value of the optimization problem when
the objective function is (\ref{3.3}), subject to constraint (\ref{3.4}) and the $\varepsilon$-relaxed version of constraint (\ref{3.5}).
$v_j^{\varepsilon}$ converges to $v_j^*$
as $\varepsilon \rightarrow 0$.
\begin{Proposition}\label{Prop3}
$\lim_{\varepsilon \downarrow 0} v_j^{\varepsilon} = v_j^*$ for every $j \in \{1,2,...,n\}$.
\end{Proposition}
\begin{proof}[Proof of Proposition 3:]
Since the constraint in Proposition \ref{Prop2} relaxes constraint (\ref{3.5}), we have $v_j^{\varepsilon} \geq v_j^*$, which implies that:
\begin{equation}\label{3.14}
    \lim\inf_{\varepsilon \downarrow 0} v_j^{\varepsilon} \geq v_j^*.
\end{equation}
The rest of the proof establishes the following inequality:
\begin{equation}\label{3.15}
    \lim\sup_{\varepsilon \downarrow 0} v_j^{\varepsilon} \leq v_j^*,
\end{equation}

The challenge is that the subset of $\Delta (\mathbf{A} \times \mathbf{B})$ that satisfies the $\varepsilon$-relaxed constraint is \textit{not convex}. My proof
constructs a distribution $\gamma \in \Delta (\mathbf{A} \times \mathbf{B})$ that respects constraints (\ref{3.4}) and (\ref{3.5}), and furthermore, is $\varepsilon$-close to a joint distribution that solves the $\varepsilon$-relaxed problem in Proposition \ref{Prop2}.
This implies that type $c_j$ sender's expected payoff is close under the two distributions.

Let $\Gamma^{\varepsilon}$ be the set of $\gamma \in \Delta (\mathbf{A} \times \mathbf{B})$ that satisfies constraint (\ref{3.4}) and the $\varepsilon$-relaxed version of constraint (\ref{3.5}). Let $\Gamma$ be the set of $\gamma \in \Delta (\mathbf{A} \times \mathbf{B})$ that satisfies constraints (\ref{3.4}) and (\ref{3.5}).
I show that for every $\eta>0$, there exists $\varepsilon>0$, such that
for every
$\gamma^{\varepsilon} \in \Gamma^{\varepsilon}$, there exists $\gamma \in \Gamma$ that is within $\eta$ away from $\gamma^{\varepsilon}$. This implies
inequality (\ref{3.15}).

First, for every $\mathbf{b} \in \mathbf{B}$, there exists a nondegenerate subset of $\Delta (\mathbf{A})$
such that $\mathbf{b}$ best replies against. Since the number of pure stage-game actions is finite, for every $\eta>0$, there exists $\varepsilon>0$, such that for every $\alpha \in \Delta (\mathbf{A})$ and $\mathbf{b} \in \mathbf{B}$
satisfying $\mathbf{b}$
is an $\varepsilon$-best reply against $\alpha$, there exists $\alpha' \in \Delta (\mathbf{A})$ within $\eta$ away from $\alpha$ such that $\mathbf{b}$ best replies against $\alpha'$.

Second, for every $\gamma^{\varepsilon} \in \Gamma^{\varepsilon}$, let
\begin{equation*}
    \mathbf{B}^* \equiv \big\{
    \mathbf{b} \in \mathbf{B}
    \big| \textrm{ }
    \mathbf{b} \textrm{ best replies against } \gamma^{\varepsilon}(\cdot|\mathbf{b})
    \big\}.
\end{equation*}
By definition, the marginal distribution of $\gamma^{\varepsilon}$ on $B$, denoted by $\beta^{\varepsilon}$, attaches probability at most $\varepsilon$ to every $\mathbf{b} \notin \mathbf{B}^*$. Consider the following modified distribution  $\gamma' \in \Delta (\mathbf{A} \times \mathbf{B})$:
\begin{itemize}
\item[1.] For every $\mathbf{b} \in \mathbf{B}^*$, there exists $\gamma^*$ that is
  $\eta$ away from $\gamma^{\varepsilon}(\cdot|\mathbf{b})$, with
$\mathbf{b}$ best replies to $\gamma^*(\cdot|\mathbf{b})$.
\item[2.] The marginal distribution of $\gamma'$ on $\mathbf{B}$ attaches probability
$\frac{\beta^{\varepsilon}(\mathbf{b})}{\beta^{\varepsilon}(\mathbf{B}^*)}$
to $\mathbf{b}$, and the distribution over $\mathbf{a}$ conditional on $\mathbf{b}$ is $\gamma^*(\cdot|\mathbf{b})$.
\end{itemize}
Since
\begin{equation*}
\sum_{(\mathbf{a},\mathbf{b}) \in \mathbf{A} \times \mathbf{B}} \gamma^{\varepsilon}(\mathbf{a},\mathbf{b})   u_s(c_1,\mathbf{a},\mathbf{b})
\leq p_h,
\end{equation*}
and $\beta^{\varepsilon}$ attaches probability less than $\varepsilon$ to every $\mathbf{b} \notin B^*$,
there exists $X: [0,1] \rightarrow \mathbb{N}$ with $\lim_{\eta \rightarrow 0} X(\eta)=0$ such that
\begin{equation}\label{3.16}
    \sum_{(\mathbf{a},\mathbf{b}) \in \mathbf{A} \times \mathbf{B}} \gamma'(\mathbf{a},\mathbf{b})   u_s(c_1,\mathbf{a},\mathbf{b})
\leq p_h+X(\eta)
\end{equation}
and
\begin{equation}\label{3.17}
  v_j^*+X(\eta) \geq
    \sum_{(\mathbf{a},\mathbf{b}) \in \mathbf{A} \times \mathbf{B}} \gamma'(\mathbf{a},\mathbf{b})   u_s(c_j,\mathbf{a},\mathbf{b})+X(\eta)
\geq \sum_{(\mathbf{a},\mathbf{b}) \in \mathbf{A} \times \mathbf{B}} \gamma^{\varepsilon}(\mathbf{a},\mathbf{b})   u_s(c_j,\mathbf{a},\mathbf{b})=v_j^{\varepsilon}.
\end{equation}
Consider two cases,
\begin{enumerate}
  \item If $\sum_{(\mathbf{a},\mathbf{b}) \in \mathbf{A} \times \mathbf{B}} \gamma'(\mathbf{a},\mathbf{b})   u_s(c_1,\mathbf{a},\mathbf{b})\leq p^h$, then $\gamma'$ satisfies constraints (\ref{3.4})
  and (\ref{3.5}), and
  attains payoff within $X(\eta)$ of $v_j^{\varepsilon}$.
  \item If $\sum_{(\mathbf{a},\mathbf{b}) \in \mathbf{A} \times \mathbf{B}} \gamma'(\mathbf{a},\mathbf{b})   u_s(c_1,\mathbf{a},\mathbf{b})> p^h$, then
  let $\gamma'' \in \Delta (A \times B)$ be
 a convex combination of $\gamma'$ and the Dirac measure on $(\mathbf{a}^L,\mathbf{b}^N)$, with  the convex weight on $\gamma'$ equals
      \begin{equation*}
        \frac{p_h}{\sum_{(\mathbf{a},\mathbf{b}) \in \mathbf{A} \times \mathbf{B}} \gamma'(\mathbf{a},\mathbf{b})   u_s(c_1,\mathbf{a},\mathbf{b})},
      \end{equation*}
 Since all types of sender's stage-game payoff is $0$ under $(\mathbf{a}^L,\mathbf{b}^N)$, $\gamma''$
satisfies constraint (\ref{3.4}). Since $\mathbf{b}^N$ best replies against $\mathbf{a}^L$, $\gamma''$ satisfies constraint (\ref{3.5}).
According to the definition of $v_j^*$, we have:
\begin{equation}\label{3.18}
   v_j^*\geq \sum_{(\mathbf{a},\mathbf{b}) \in \mathbf{A} \times \mathbf{B}} \gamma''(\mathbf{a},\mathbf{b})   u_s(c_j,\mathbf{a},\mathbf{b}) .
\end{equation}
According to (\ref{3.16}) and (\ref{3.17}), we also have:
\begin{equation*}
\sum_{(\mathbf{a},\mathbf{b}) \in \mathbf{A} \times \mathbf{B}} \gamma''(\mathbf{a},\mathbf{b})   u_s(c_j,\mathbf{a},\mathbf{b}) \geq
  \frac{p_h}{p_h+X(\eta)} \Big( \sum_{(\mathbf{a},\mathbf{b}) \in \mathbf{A} \times \mathbf{B}} \gamma'(\mathbf{a},\mathbf{b})   u_s(c_j,\mathbf{a},\mathbf{b}) -X(\eta)\Big)
  \end{equation*}
        \begin{equation}\label{3.19}
\geq
\frac{p_h}{p_h+X(\eta)} \Big(v_j^{\varepsilon}-X(\eta)\Big).
      \end{equation}
The expression on the RHS of (\ref{3.19}) implies that for every $\rho>0$, there exists $\eta>0$ such that once we pick $\varepsilon$ according to $\eta$, we have:
\begin{equation*}
    v_j^* \geq \frac{p_h}{p_h+X(\eta)} \Big(v_j^{\varepsilon}-X(\eta)\Big)\geq v_j^{\varepsilon}-\rho.
\end{equation*}
\end{enumerate}
This leads to (\ref{3.15}). Inequalities (\ref{3.14}) and (\ref{3.15}) together imply Proposition \ref{Prop3}.
\end{proof}

\subsection{Tightness of Payoff Upper Bound: Equilibrium Construction}\label{sub4.3}
Let
\begin{equation*}
    v^N \equiv \Big(-c_1(1-p_h),-c_2(1-p_h),...,-c_n (1-p_h) \Big),
\end{equation*}
\begin{equation*}
 v^L \equiv \Big(p_h+(1-c_1)(1-p_h), p_h+(1-c_2)(1-p_h),...,p_h+(1-c_n)(1-p_h) \Big),
\end{equation*}
and $v^H \equiv (p_h,p_h,...,p_h)$, which are the sender's stage-game payoffs from pure stage-game action profiles $(\mathbf{a}^L,\mathbf{b}^N)$, $(\mathbf{a}^L,\mathbf{b}^T)$, and $(\mathbf{a}^H,\mathbf{b}^T)$, respectively.
The receiver has an incentive to play $\mathbf{b}^T$ against $\rho \mathbf{a}^L +(1-\rho) \mathbf{a}^H$ if and only if $\rho \leq \rho^*$.
For every $\rho \in [0,\rho^*]$, let
\begin{equation}\label{3.20}
    v(\rho) \equiv \frac{(1-\rho) c_1}{\rho (1-c_1) +c_1} v^H
    + \frac{\rho c_1}{\rho (1-c_1) +c_1} v^L
    +\frac{\rho (1-c_1)}{\rho (1-c_1) +c_1} v^N.
\end{equation}
One can verify that $v(0)=v^H$ and $v(\rho^*)=(v_1^*,...,v_n^*)$.
The second statement of Theorem \ref{Theorem1} is implied by the following proposition:
\begin{Proposition}\label{Prop4}
For every $\varepsilon>0$ and $\rho \in [0,\rho^*)$, there exists $\underline{\delta} \in (0,1)$ such that for every $\pi \in \Delta(\mathcal{C})$ with $\pi_1 \geq \varepsilon$ and $\delta>\underline{\delta}$, there exists an equilibrium in which the sender's payoff is $v(\rho)$.
\end{Proposition}
I provide a constructive proof in Appendix \ref{secA}.
In the remainder of this section, I explain the ideas behind the construction in an example with \textit{two types}, i.e., $\mathcal{C} \equiv \{c_1,c_2\}$. The conceptual challenge is to let a patient sender reveal his persistent private information for unbounded number of periods while
preserving his informational advantage. I also explain how to square my result with the payoff upper bound in Fudenberg, Kreps and Maskin (1990).

\paragraph{Preliminaries:} In the constructed equilibrium, the stage-game outcome at every on-path history is a distribution supported on $\big\{(\mathbf{a}^H,\mathbf{b}^T),(\mathbf{a}^L,\mathbf{b}^T),(\mathbf{a}^L,\mathbf{b}^N)\big\}$.
The sender's continuation value at every on-path history is a convex combination of $v^H$, $v^L$, and $v^N$, and belongs to a polytope $V^*$ with the following four vertices:
 $v^H$, $v^*$, $\overline{v} \equiv q^* (\rho^* v^L+(1-\rho^*) v^H) +(1-q^*) v^N$ with $q^* \in [0,1]$ pinned down by the condition that the first entry of the above vector equals $0$,
 and
 $\underline{v} \equiv p^* v^H +(1-p^*) v^N$ with
 \begin{equation}
 p^* \equiv  \frac{c_1(1-p_h)}{p_h+c_1(1-p_h)}.
 \end{equation}
In an environment with \textit{two types}, $V^*$ is depicted as the yellow set in Figure 1.

\paragraph{State Variables \& Phases:} When there are two types, I keep track of two state variables:\footnote{When there are three or more types, one needs to keep track of two additional state variables, see Appendix A.}
\begin{itemize}
  \item[1.] The probability of type $c_1$ sender in the receiver's posterior belief, denoted by $\eta(h^t)$. I call this the sender's \textit{reputation} at $h^t$.
  The initial value of $\eta(\cdot)$ is $\pi_1$.
  \item[2.] The sender's continuation value, denoted by $v(h^t) \in \mathbb{R}^m$. The initial value of $v(\cdot)$ is $v(\rho)$, namely, the target payoff. Given that $v(h^t)$ is a convex combination of $v^H$, $v^N$ and $v^L$, it is equivalent to keep track of their convex weights, denoted by $p^H(h^t)$, $p^N(h^t)$, and $p^L(h^t)$.
\end{itemize}
The equilibrium consists of three phases.
\begin{enumerate}
  \item Play starts from \textit{an active learning phase} in which the receiver plays $\mathbf{b}^T$ and the two types of the sender mix between $\mathbf{a}^H$ and $\mathbf{a}^L$ in \textit{most of the periods}, with type $c_1$ playing $\mathbf{a}^H$ with higher probability compared to type $c_2$.
 An important exception is
      when the receiver's posterior belief attaches probability close to $1$ to the sender being type $c_1$, in which case type $c_1$ mixes between $\mathbf{a}^H$ and $\mathbf{a}^L$, and type $c_2$ plays $\mathbf{a}^L$ for sure.\footnote{Another exception is when $p^L(h^t)$ is strictly between $0$ and $1-\delta$, in which type $c_1$ plays $\mathbf{a}^H$ for sure and type $c_2$ potentially mixes between $\mathbf{a}^H$ and $\mathbf{a}^L$. The details are described under Class 2 histories in Appendix A.}

  \item When the sender's continuation value is close to his minmax payoff, play enters a \textit{rebounding phase} in which learning \textit{temporarily stops}. In this phase, the receiver plays $\mathbf{b}^N$ and all types of sender plays $\mathbf{a}^L$. Play transits from the rebounding phase back to the active learning phase when the sender's continuation value is high enough such that
      no type of the sender's continuation value falls below his minmax payoff when the state in the next period is $l$ and the sender recommends action $H$.

      This phase is required since
      at some on-path histories,
      type $c_1$ sender's continuation value approaches his minmax payoff $0$
      and other types' continuation values cannot be delivered in equilibria without learning.
      The rebounding phase offers a solution to such dilemma by temporarily stops learning and prescribes the low-payoff outcome for several periods. This increases all types of sender's continuation values, while does not change the ratio between the convex weight of $v^L$ and the convex weight of $v^H$ in the sender's continuation value.

  \item Play enters an \textit{absorbing phase} in which learning about the sender's type stops forever, and the continuation play consists only of $(\mathbf{a}^L,\mathbf{b}^N)$ and $(\mathbf{a}^H,\mathbf{b}^T)$. Play reaches this phase either after the sender reveals his type, or after he has played $\mathbf{a}^L$ too frequently in the active learning phase. Despite the sender can flexibly choose whether to play $\mathbf{a}^H$ or $\mathbf{a}^L$ in the active learning phase, his action choices affect the time at which play reaches the absorbing phase, and his continuation value after play enters the absorbing phase. For example, if he lies frequently, then play reaches the absorbing phase sooner after which he receives a low continuation payoff.
\end{enumerate}
\paragraph{Benefit from Persistent Private Information:} I provide intuition for why the above construction enables type $c_2$ sender to extract information rent in the long run and obtain discounted average payoff close to $v_2^*$ (i.e., strictly above $p_h$) when $\delta$ is arbitrarily close to $1$.

I start from reviewing the argument in Fudenberg, Kreps and Maskin (1990), which explains why the sender's payoff cannot exceed $p_h$ when there is \textit{only one type}.
At every $h^t$ where the receiver plays $H$ with positive probability, there exists a message $m'$ such that the receiver plays $L$ for sure after observing $m'$ at $h^t$, and the sender sends $m'$ with strictly positive probability at $h^t$ when the state is $l$.
Therefore, the following strategy is the sender's best reply against the receiver's equilibrium strategy, under which the sender's payoff \textit{in each period} is no more than $p_h$:
\begin{itemize}
  \item at every $h^t$ where the receiver plays $H$ with positive probability,
  send message $m'$ with probability $1$ when the state is $l$.
\end{itemize}

Next, I explain why type $c_2$ can obtain  payoff higher than $p_h$ in the stage game when there are two types.
The above argument breaks down since the receiver may have an incentive to play $\mathbf{b}^T$ at histories where type $c_2$ sender plays $\mathbf{a}^L$ for sure. This requires type $c_1$ to play $\mathbf{a}^H$ with high enough probability, which distinguishes his behavior from that of type $c_2$'s. Therefore, the state realization and the sender's message are informative signals about the sender's type.

Then, I explain how type $c_2$ can reveal information and extract rent \textit{in the long run}.
Following the argument in Fudenberg, Kreps and Maskin (1990), if at every history, type $c_2$ sender plays the pure strategy that minimizes his stage-game payoff among the ones in the support of his equilibrium strategy, then this modified strategy is type $c_2$ sender's equilibrium best reply,
which I denote by $\widetilde{\sigma}_{c_2}$ and from which type $c_2$ obtains his equilibrium payoff.

If type $c_2$ plays according to $\widetilde{\sigma}_{c_2}$, then
his expected stage-game payoff exceeds $p_h$ only at histories where his \textit{equilibrium strategy} prescribes $\mathbf{a}^L$ with probability $1$ while the receiver plays $\mathbf{b}^T$ with positive probability. His reputation changes gradually over time following the pattern of a \textit{cycle}, with
outcome $(\mathbf{a}^L,\mathbf{b}^T)$ occurs only when the receiver's belief attaches probability close to $1$ to the sender being type $c_1$. This high-payoff outcome can arise in \textit{unboundedly many periods} since type $c_2$ can \textit{rebuild his reputation} in periods where he plays $\mathbf{a}^H$. Such reputation rebuilding is feasible given that type $c_1$ mixes between $\mathbf{a}^H$ and $\mathbf{a}^L$ in periods where active learning takes place. The above observation
 \textit{does not} contradict the martingale property of beliefs since the receiver's belief is updated based on the sender's equilibrium strategy, and therefore, her belief is \textit{not necessarily a martingale} under type $c_2$ sender's modified best reply $\widetilde{\sigma}_{c_2}$.

Lastly, I provide a \textit{heuristic derivation} of the maximal frequency with which the sender can extract information rent from a learning perspective, which complements the formal proofs of Propositions \ref{Prop2} and \ref{Prop3} using occupation measures. It helps to understand why my construction works only when $\rho<\rho^*$.
In this heuristic derivation, I restrict attention to situations in which the sender's stage-game action is supported on $\{\mathbf{a}^H,\mathbf{a}^L\}$, and the receiver's stage-game action is supported on $\{\mathbf{b}^T,\mathbf{b}^N\}$. Therefore, the receiver \textit{cannot} learn the sender's type when the realized state is $h$, but can potentially learn about the sender's type when the realized state is $l$.

To start with,
the receiver's willingness to play $\mathbf{b}^T$ implies that $\mathbf{a}^H$ needs to be played with probability at least $1-\rho^*$. Since belief is a martingale, this provides an upper bound on the \textit{relative speed of learning}, which measures the magnitude with which the sender's reputation improves after reporting the low state honestly, and the magnitude with which his reputation deteriorates after misstating that the state is high: \begin{equation}\label{4.23}
    \frac{\eta(h^t,(l,l))-\eta(h^t)}{\eta(h^t,(l,h))-\eta(h^t)} \leq \frac{\rho^*}{1-\rho^*}.
\end{equation}
A belief updating rule that meets the above requirement is given by:
  \begin{equation}\label{4.21}
    \eta(h^t,(l,h))-\eta^*=
    (1-\lambda (1-\rho^*))
    (\eta(h^t)-\eta^*),
\end{equation}
and
\begin{equation}\label{4.22}
        \eta(h^t,(l,l))-\eta^*=\min\Big\{1-\eta^*,
    (1+\lambda \rho^*)
    (\eta(h^t)-\eta^*)\Big\},
\end{equation}
where $\eta^* \in (0,\eta(h^0))$ is a belief lower bound in the active learning phase, and $\lambda>0$ is a parameter that measures of speed of receiver-learning. Equations  (\ref{4.21}) and (\ref{4.22}) \textit{pin down} the probabilities with which each type of sender plays $\mathbf{a}^H$ and $\mathbf{a}^L$.

For every on-path history $h^t$ such that the sender's reputation has not reached one, and let $N$ be the number of periods in which the state realization is $l$. If the sender
plays $\mathbf{a}^H$ and $\mathbf{a}^L$ with frequencies $1-\rho$ and $\rho$ in these $N$ periods, the receiver's posterior belief attaches  probability
\begin{equation}\label{4.24}
\eta^*+ (\eta(h^0)-\eta^*)  \Big(   (1+\lambda \rho^*)^{1-\rho}    (1-\lambda (1-\rho^*))^{\rho}  \Big)^N,
\end{equation}
to the sender being type $c_1$. This
is no less than $\eta(h^0)$ \textit{if and only if}
\begin{equation}\label{4.25}
    (1+\lambda \rho^*)^{1-\rho}    (1-\lambda (1-\rho^*))^{\rho} \geq 1.
\end{equation}
Applying the Taylor's theorem, there exists $\lambda>0$ under which (\ref{4.25}) holds \textit{if and only if}
$\rho<\rho^*$.
Therefore, if type $c_2$ sender's frequency of playing $\mathbf{a}^L$ exceeds $\rho^*$, then his reputation deteriorates in the long run and extracting information rent at this rate is \textit{not sustainable}. If the sender's frequency of playing $\mathbf{a}^L$ is strictly below $\rho^*$, then there exists $\overline{\lambda}>0$ such that when $\lambda <\overline{\lambda}$,
his reputation improves over time, which allows type $c_2$ to sustain his information rents. This provides a learning explanation for constraint (\ref{3.5}), which requires that the relative frequency between outcomes $(\mathbf{a}^H,\mathbf{b}^T)$ and $(\mathbf{a}^L,\mathbf{b}^T)$ cannot fall below $\frac{1-\rho^*}{\rho^*}$.

\section{The Effects of Ethical-Type Sender}\label{sec5}
This section relaxes the assumption in the baseline model that $c_1<1$ and examines whether the presence of \textit{ethical-type senders}, i.e., types with lying costs no less than $1$, enables non-ethical types to attain their optimal commitment payoffs.

Formally, let $c \in \mathcal{C} \equiv \{c_1,...,c_n\}$, with $0 \leq c_n < c_{n-1}<...<c_2<c_1$.
Let $\pi \in \Delta (\mathcal{C})$ be the receivers' prior belief.
I focus on cases in which there exists $c \in \mathcal{C}$ such that $c \in [0,1)$, i.e., there exists a type of sender
who is non-ethical and his optimal disclosure policy is
stochastic.

 I introduce two constrained optimization problems for every $c \in \mathcal{C}\cap [1,+\infty)$, and later provide economic interpretations. Let
\begin{equation}\label{5.1}
    \overline{v} (c) \equiv \max_{\gamma \in \Delta (\mathbf{A} \times \mathbf{B})} \sum_{(\mathbf{a},\mathbf{b}) \in \mathbf{A} \times \mathbf{B}} \gamma (\mathbf{a},\mathbf{b}) u_s(c,\mathbf{a},\mathbf{b}),
\end{equation}
subject to the constraint that there exists $c' \in \mathcal{C} \cap [0,1)$ such that:
\begin{equation}\label{5.2}
\sum_{(\mathbf{a},\mathbf{b}) \in \mathbf{A} \times \mathbf{B}} \gamma (\mathbf{a},\mathbf{b}) u_s(c',\mathbf{a},\mathbf{b})
\geq p_h +\rho^* (1-p_h) (1-c').
\end{equation}
Let
\begin{equation}\label{5.3}
    \underline{v}(c) \equiv \min_{\gamma \in \Delta (\mathbf{A} \times \mathbf{B})} \sum_{(\mathbf{a},\mathbf{b}) \in \mathbf{A} \times \mathbf{B}} \gamma (\mathbf{a},\mathbf{b}) u_s(c,\mathbf{a},\mathbf{b})
\end{equation}
subject to:
\begin{equation}\label{5.4}
    \sum_{(\mathbf{a},\mathbf{b}) \in \mathbf{A} \times \mathbf{B}} \gamma (\mathbf{a},\mathbf{b}) u_s(c_1,\mathbf{a},\mathbf{b})
\geq 0
\end{equation}
and for  every $\mathbf{b} \in \mathbf{B}$ that
the marginal distribution of
$\gamma$ on $\mathbf{B}$ attaches positive probability to,
\begin{equation}\label{5.5}
 \mathbf{b} \in \arg\max_{\mathbf{b'} \in \mathbf{B}}   u_r(\gamma(\cdot|\mathbf{b}),\mathbf{b'}).
\end{equation}
Recall the definition of type $c_j$'s optimal commitment payoff $v_j^{**}$ in (3.2).
\begin{Theorem}\label{Theorem3}
Suppose $\pi$ has full support and $\mathcal{C} \cap [0,1)$ is nonempty,
\begin{enumerate}
  \item if there exists no $c \in \mathcal{C} \cap [1,+\infty)$ such that $\overline{v}(c) > \underline{v} (c)$, then there exist $\eta>0$ and $\underline{\delta} \in (0,1)$, such that in every BNE when $\delta>\underline{\delta}$, every type $c_j \in \mathcal{C} \cap [0,1)$ obtains payoff  no more than $v_j^{**}-\eta$.
  \item if there exists $c \in \mathcal{C} \cap [1,+\infty)$ such that $\overline{v}(c) \geq \underline{v} (c)$, then
  for every $\varepsilon>0$, there exists $\underline{\delta} \in (0,1)$, such that when $\delta >\underline{\delta}$,
  there exists a sequential equilibrium in which every type $c_j \in \mathcal{C} \cap [0,1)$ obtains payoff at least $v_j^{**}-\varepsilon$.
\end{enumerate}
\end{Theorem}
The proof is in Appendix \ref{secB} (statement 1) and Appendix \ref{secC} (statement 2).
Theorem \ref{Theorem3} provides a \textit{necessary and sufficient condition}
for the existence of equilibrium in which non-ethical-type sender can attain his optimal commitment payoff.
The attainability of this payoff depends only on
the existence of
an \textit{ethical-type sender} whose lying cost $c$ satisfies  $\overline{v}(c) \geq \underline{v} (c)$. When $\mathcal{C} \cap [1,+\infty)$ is empty,
the second statement of Theorem \ref{Theorem3} applies, which is implied by
Theorem \ref{Theorem1} given that $v_j^*<v_j^{**}$ for every $c_j \in \mathcal{C} \cap [0,1)$.
When $\mathcal{C} \cap [1,+\infty)$ is nonempty, namely,
at least one type of the sender is ethical, the following predictions emerge which are different from the ones in models without ethical types:
\begin{enumerate}
  \item the attainability of optimal commitment payoffs depends not only on $c_1$ but also on \textit{other ethical types}, in particular, the \textit{lowest lying cost} among these ethical types;
  \item an increase the value of $c_1$ can \textit{decrease} other types of sender's highest equilibrium payoffs.
\end{enumerate}

To understand these implications, I provide economic interpretations for $\overline{v}(c)$ and $\underline{v}(c)$, stated as Lemma \ref{L5.1} and Lemma \ref{L5.2}.
Let $c$ be a typical element of $\mathcal{C} \cap [1,+\infty)$. Lemma \ref{L5.1} shows that
in every equilibrium where at least one type of non-ethical sender obtains his optimal commitment payoff,
$\overline{v}(c)$ is an upper bound for
type $c$ ethical sender's equilibrium payoff.
\begin{Lemma}\label{L5.1}
For every small enough $\varepsilon>0$, there exists $\underline{\delta} \in (0,1)$ such that when $\delta >\underline{\delta}$, in every equilibrium in which
there exists $c_j \in \mathcal{C} \cap [0,1)$ such that
type $c_j$ attains payoff more than $v_j^{**}-\varepsilon$,
there exists $c \in \mathcal{C} \cap [1,+\infty)$ such that
type $c$'s payoff is no more than $\overline{v}(c)+\varepsilon$.
\end{Lemma}
The proof is in Appendix \ref{subB.1}. This is because every ethical-type sender's stage-game payoff is maximized by reporting the state truthfully while the receiver's action matches the state.\footnote{I am ruling out the receiver's strategy of taking the high action with ex ante probability $1$.} But some ethical type needs to lie with positive probability and suffer from payoff losses in order to provide cover for the nonethical types so that the latter can attain their optimal commitment payoffs.
The first constrained optimization problem
characterizes the minimal amount of payoff loss type $c$ needs to incur in order for a non-ethical type to obtain his optimal commitment payoff.

The next lemma shows that
$\underline{v}(c)$ is a lower bound for type $c$ sender's equilibrium payoff, given that he has the option to
adopt type $c_1$ sender's equilibrium strategy, from which the latter's payoff is no less than his minmax payoff $0$.
\begin{Lemma}\label{L5.2}
For every $\varepsilon>0$, there exists $\underline{\delta} \in (0,1)$ such that when $\delta >\underline{\delta}$, for every $c \in \mathcal{C} \cap [1,+\infty)$,
type $c$'s payoff in every equilibrium is at least $\underline{v}(c)-\varepsilon$.
\end{Lemma}
The proof is in Appendix \ref{subB.2}. Solving for the values of $\overline{v}(c)$ and $\underline{v}(c)$, and using the fact that
$\rho^* \equiv \frac{p_h}{1-p_h}$,
we obtain:
\begin{equation}\label{5.6}
    \overline{v}(c) = p_h-\rho^* (c-1) (1-p_h)=p_h(2-c),
\end{equation}
and
\begin{equation}\label{5.7}
    \underline{v}(c) = \frac{(c_1-c) \rho^* p_h (1-p_h)}{p_h +\rho^* c_1(1-p_h)}=\frac{p_h(c_1-c)}{1+c_1}.
\end{equation}
The necessary and sufficient condition in Theorem \ref{Theorem3} translates into:
\begin{equation}\label{5.8}
  c_1 (c-1) \leq 2.
\end{equation}
Let $c^* \equiv \min \Big\{
\mathcal{C} \cap [1,+\infty)
\Big\}$.
Theorem \ref{Theorem3} and the above derivations lead to the following corollary:
\begin{Corollary}
Suppose $\pi$ has full support and neither $\mathcal{C} \cap [0,1)$
nor
$\mathcal{C} \cap [1,+\infty)$ is empty,
\begin{enumerate}
  \item if $c_1(c^*-1) \leq 2$, then
  for every $\varepsilon>0$, there exists $\underline{\delta} \in (0,1)$, such that when $\delta >\underline{\delta}$,
  there exists a sequential equilibrium in which every type $c_j \in \mathcal{C} \cap [0,1)$ obtains payoff at least $v_j^{**}-\varepsilon$.
  \item if $c_1 (c^*-1)>2$, then there exist $\eta>0$ and $\underline{\delta} \in (0,1)$, such that in every BNE when $\delta>\underline{\delta}$, every type $c_j \in \mathcal{C} \cap [0,1)$ obtains payoff  no more than $v_j^{**}-\eta$.
\end{enumerate}
\end{Corollary}
Corollary 2 suggests that the attainability of non-ethical types' optimal commitment payoffs depends only on
the \textit{highest} and the \textit{lowest lying cost among the ethical types}.
Aside from the knife-edge case in which $1 \in \mathcal{C}$,
the optimal commitment payoffs are attainable if and only if the highest cost of lying $c_1$ is below some cutoff.\footnote{If $1 \in \mathcal{C}$, namely, there exists a type whose lying cost is exactly $1$, then all types of non-ethical senders can approximately attain their optimal commitment payoffs, irrespective of what the highest lying cost is.} For example, if there is only one ethical type in the support of receivers' prior belief, then the optimal commitment payoffs are attainable if and only if $ c=c_1 \in [1,2]$.\footnote{When a type of sender's lying cost is $2$, his payoff from the non-ethical sender's optimal disclosure policy is $0$ upon receiving the receiver's trust.}

Different from models without ethical types, non-ethical type sender's highest attainable payoff
strictly \textit{decreases} with the highest lying cost when the latter is large enough.
In a model with at least one ethical type (call it the \textit{original ethical type}),
introducing an additional ethical type who has a high cost of lying can \textit{reduce} non-ethical types' highest equilibrium payoffs.
This contrasts to models \textit{without} ethical types in which every type's highest equilibrium payoff is non-decreasing after introducing an additional type, no matter whether this additional type is ethical or not.

Intuitively, this is driven by a novel \textit{outside option effect}, which is absent in models without ethical types or in models with commitment types. In particular, once receivers entertain the possibility that the sender can be a type who incurs a high lying cost, the original ethical type enjoys weakly better outside options since he can imitate the equilibrium strategy of the highest-cost type, from which the latter receives at least his minmax payoff. This improvement in his outside option restricts the frequency with which the original ethical type can lie on the equilibrium path, which in turn reduces the non-ethical types' long-term payoff gain from reputation building and milking, since the number of times with which he can lie and pool with some ethical type decreases.

\section{Conclusion}\label{sec6}
This paper studies repeated communication games in which the sender has persistent private information about his psychological cost of lying.
I characterize every type of patient sender's highest equilibrium payoff as the solution of a constrained optimization problem, and provide a necessary and sufficient condition for every type of non-ethical sender to attain his optimal commitment payoff. My results also clarify the distinction between a strategic-type sender who faces high lying cost, and a commitment-type sender
who mechanically adopts his Bayesian disclosure policy
in every period.
In addition, the possibility of being ethical and having a high-lying cost can hurt non-ethical senders due to a novel outside option effect.

\appendix
\section{Proof of Proposition 4}\label{secA}
For $\delta$ close enough to $1$, I construct an equilibrium such that at every on-path history, the distribution over
players' stage-game action profiles are supported on $\big\{(\mathbf{a}^H,\mathbf{b}^T),(\mathbf{a}^L,\mathbf{b}^T),(\mathbf{a}^L,\mathbf{b}^N)\big\}$, and the sender's continuation payoff is a convex combination of $v^H$, $v^N$ and $v^L$, belonging to a polytope $V^*$ with the following four vertices:
 $v^H$, $v^*$, $\overline{v} \equiv q^* (\rho^* v^L+(1-\rho^*) v^H) +(1-q^*) v^N$ with $q^* \in [0,1]$ pinned down by the condition that the first entry of the above vector equals $0$,
 and
 $\underline{v} \equiv p^* v^H +(1-p^*) v^N$ with
 \begin{equation}
 p^* \equiv  \frac{c_1(1-p_h)}{p_h+c_1(1-p_h)}.
 \end{equation}
Since $c_1>c_j$ for every $j \neq 1$, we know that $v^* > \overline{v} \vee v^H $, and $\overline{v} \wedge v^H \geq \underline{v}$, and the second inequality
 is strict for all entries of this $m$-dimensional vector except for the first entry.

In what follows,
I verify that the sender's continuation payoff vector at \textit{every history} belongs to $V^*$ after describing players' equilibrium strategies and the evolution of their continuation values.
\subsection{State Variables \& Useful Constants}
The constructed equilibrium keeps track of the following three sets of state variables:
\begin{itemize}
  \item[1.] The probability of type $c_1$ sender in the receiver's posterior belief, denoted by $\eta(h^t)$. I call this the sender's \textit{reputation} at $h^t$.
  \item[2.] The promised continuation payoff to the sender, denoted by $v(h^t) \in \mathbb{R}^n$. Given that it is a convex combination of $v^H$, $v^N$ and $v^L$, it is equivalent to keep track of the convex weights of $v^H$, $v^N$ and $v^L$ in $v(h^t)$, which I denote by $p^H(h^t)$, $p^N(h^t)$, and $p^L(h^t)$, respectively.
  \item[3.] The lowest-cost type in the support of receivers' posterior belief, denoted by $\underline{c}(h^t)$, as well as
  its probability according to the receiver's belief at $h^t$.
\end{itemize}
The initial values of these state variables are $\eta(h^0)=\pi_1$, $p^H(h^0)=\frac{(1-\rho) c_1}{\rho (1-c_1) +c_1}$, $p^L(h^0)=\frac{\rho c_1}{\rho (1-c_1) +c_1}$, $p^N(h^0) = \frac{\rho (1-c_1)}{\rho (1-c_1) +c_1}$, $\underline{c}(h^0)=c_n$, and the probability of type $\underline{c}(h^0)$ is $\pi_n$.

I construct some constants for future reference. First, let
\begin{equation*}
    \widehat{\delta} \equiv \delta \frac{1-p_h}{1-\delta p_h},
\end{equation*}
which is called the sender's \textit{effective discount factor}
and is strictly smaller than $\delta$.
One can verify that for any $p_h \in (0,1)$, the effective discount factor converges to $1$ as $\delta \rightarrow 1$.
I will replace $\delta$ with $\widehat{\delta}$ when specifying the evolution of the sender's continuation payoff vector.

Recall that $\pi \equiv (\pi_1,...,\pi_n)$ is the receiver's prior belief about the sender's lying cost.
For every $j \geq 3$, let
$k_j$ be the smallest integer $k \in \mathbb{N}$ such that:
\begin{equation}\label{A.16}
   \Big( 1-(1-\rho^*) \pi_1 \Big) \frac{(\pi_j/k)}{\sum_{\tau=2}^{j-1} \pi_{\tau} +(\pi_j/k)}
    \leq \rho^* \equiv \frac{p_h}{1-p_h}.
\end{equation}
Let $\eta^* \in [(1-\rho^*) \pi_1,\pi_1)$ be large enough such that for every $\eta \in [\eta^*,\pi_1]$, we have:
\begin{equation}\label{A.17}
    \frac{\pi_1-\eta}{\pi_1 (1-\eta)}
    \leq \min_{j \in \{3,...,n\}}
    \Big\{
    \frac{\pi_j/k_j }{\pi_2+...+\pi_j}
    \Big\}
\end{equation}
The ensuing construction implies that $\eta^*$ is a uniform lower bound on the probability of type $c_1$
that applies to all
histories where active learning takes place.

For every fixed
$\rho \in (0,\rho^*)$, there exists a rational number $\widehat{n}/\widehat{l} \in (\rho,\rho^*)$ with $\widehat{n},\widehat{l} \in \mathbb{N}$. Moreover, there exists an integer $j \in \mathbb{N}$ such that
\begin{equation*}
  \frac{\widehat{n}}{\widehat{l}} =\frac{\widehat{n}j}{\widehat{l}j} >  \frac{\widehat{n}j}{\widehat{l}j+1} > \rho.
\end{equation*}
Let $n \equiv \widehat{n}j$ and $l \equiv \widehat{l}j$. Let
\begin{equation}\label{A.2}
    \widetilde{\rho} \equiv
       \frac{1}{2} \Big(
        \frac{n}{l}+\frac{n}{l+1}
        \Big),
\end{equation}
and
\begin{equation}\label{A.3}
    \widehat{\rho} \equiv       \frac{1}{2} \Big(
        \frac{n}{l}+\rho^*
        \Big).
\end{equation}
Let $\underline{\delta} \in (0,1)$ to be large enough such that for every $\delta >\underline{\delta}$,
\begin{equation}\label{A.4}
        \frac{\delta+\delta^2+...+\delta^n}{\delta+\delta^2+...+\delta^l} >
        \widetilde{\rho}
         >\frac{\delta^{l-n+1}(\delta+\delta^2+...+\delta^n)}{\delta+\delta^2+...+\delta^{l+1}}.
\end{equation}
I require $\delta$ to be large enough such that its corresponding $\widehat{\delta}>\underline{\delta}$. In the subsequent proof, I will introduce additional requirements on $\delta$, which are compatible given that all of them require $\delta$ to be large enough (but strictly smaller than $1$).
By construction, $\rho^* >\widehat{\rho}>\frac{n}{l}> \widetilde{\rho} >\frac{n}{l-1}>\rho$.
Let $\lambda>0$ to be small enough such that:
\begin{equation}\label{A.4}
 (1+\lambda \rho^*)^{1-\widehat{\rho}}    (1-\lambda (1-\rho^*))^{\widehat{\rho}} > 1.
\end{equation}
The existence of such $\lambda>0$ is implied by the Taylor's expansion given that $\rho^*>\widehat{\rho}$.
This will later be related to the speed of receiver-learning.

\subsection{Equilibrium Strategies \& Continuation Values}
Play consists of three phases: an active learning phase, an absorbing phase, and a rebounding phase.
I partition the set of active learning phase histories into two subsets: Class 1 histories and Class 2 histories, depending on the magnitude of $p^L(h^t)$.

\paragraph{Class 1 Active Learning  History:}
Play starts from a Class 1 history of the \textit{active learning phase}. A history belongs to this class if and only if:
\begin{enumerate}
  \item $p^L(h^t) \geq 1-\widehat{\delta}$,
  \item The first entry of the following $m$-dimensional vector is non-negative:
\begin{equation}\label{A.5}
\frac{p^L(h^t)-(1-\widehat{\delta})}{\widehat{\delta}} v^L
+\frac{p^H(h^t)}{\widehat{\delta}} v^H
+\frac{p^N(h^t)}{\widehat{\delta}} v^N.
\end{equation}
\end{enumerate}
In this phase, the receiver plays $\mathbf{b}^T$ (i.e., trust the sender's recommendation). The sender of types $c_2$ to $c_n$ play the same mixed action, and type $c_1$ sender plays differently. All types send message $h$ with probability $1$ if $\omega_t=h$, i.e., $\eta(h^t,(h,h))=\eta(h^t)$.
The mixing probabilities of type $c_1$ as well as other types when the state is $l$ is pinned down by the following belief-updating formulas:
\begin{equation}\label{A.2}
    \eta(h^t,(l,h))-\eta^*=
    (1-\lambda (1-\rho^*))
    (\eta(h^t)-\eta^*),
\end{equation}
and
\begin{equation}\label{A.3}
        \eta(h^t,(l,l))-\eta^*=\min\Big\{1-\eta^*,
    (1+\lambda \rho^*)
    (\eta(h^t)-\eta^*)\Big\},
\end{equation}
with
\begin{itemize}
  \item $\eta^*$ is the constant defined in (\ref{A.17}), which is the lower bound on the receiver's belief in the active learning phase,
  \item $\lambda>0$ is a constant that measures the speed of learning and satisfies (\ref{A.4}),
  \item $\eta(h^t,(h,h))$ is the receiver's belief in period $t+1$ when $\omega_t=m_t=h$,
  \item $\eta(h^t,(l,h))$ is the receiver's belief in period $t+1$ when $\omega_t=l$ and $m_t=h$,
  \item $\eta(h^t,(l,l))$ is the receiver's belief in period $t+1$ when $\omega_t=m_t=l$.
\end{itemize}
Given that on the equilibrium path, the sender does not send message $l$ when $\omega_t=h$, $\eta(h^t,(h,l))$ can be arbitrary. I do not specify the receiver's belief since the sender's continuation value at $\eta(h^t,(h,l))$ equals $\underline{v}$, which can be implemented under \textit{any belief} of the receiver's.

According to  (\ref{A.2}) and (\ref{A.3}) as well as the martingale property of the receiver's beliefs, the sender plays $\mathbf{a}^H$ with probability at least $1-\rho^*$ in every period, which implies that conditional on observing message $h$, the receiver's posterior belief attaches probability at least $1/2$ to $\omega_t=h$, and hence,
has an incentive to play $\mathbf{b}^T$.

Next, I construct the evolution of the sender's continuation value.
If $h^t$ is such that $\eta(h^t,(l,l))<1$, then
the sender's continuation payoff is given by:
\begin{equation}\label{A.6}
    v(h^t,(h,h))=v(h^t),
\end{equation}
\begin{equation}\label{A.7}
    v(h^t,(l,h))= \frac{p^H(h^t)}{\widehat{\delta}} v^H
    +\frac{p^L(h^t)-(1-\widehat{\delta})}{\widehat{\delta}} v^L
    +\frac{p^N(h^t)}{\widehat{\delta}} v^N,
\end{equation}
\begin{equation}\label{A.8}
    v(h^t,(l,l))=\frac{p^H(h^t)-(1-\widehat{\delta})}{\widehat{\delta}} v^H
    +\frac{p^L(h^t)}{\widehat{\delta}} v^L
    +\frac{p^N(h^t)}{\widehat{\delta}} v^N,
\end{equation}
and $v(h^t,(h,l))=\underline{v}$.
Under the continuation values specified in (\ref{A.6}), (\ref{A.7}) and (\ref{A.8}), each type of sender is indifferent between sending message $h$ and sending message $l$ in period $t$ when $\omega_t=l$, and strictly prefers to send message $h$ when $\omega_t=h$.

If $h^t$ is such that $\eta(h^t,(l,l))=1$, then
the sender's continuation payoff is $v(h^t,(h,h))=v(h^t)$, $v(h^t,(h,l))=\underline{v}$,
\begin{equation*}
    v(h^t,(l,h))= \frac{p^H(h^t)}{\widehat{\delta}} v^H
    +\frac{p^L(h^t)-(1-\widehat{\delta})}{\widehat{\delta}} v^L
    +\frac{p^N(h^t)}{\widehat{\delta}} v^N,
\end{equation*}
and
\begin{equation}\label{A.10}
    v(h^t,(l,l))= q(h^t) v^H +(1-q(h^t)) v^N,
\end{equation}
where $q(h^t) \in [0,1]$ is such that the first entry of $v(h^t,(l,l))$ equals the first entry of vector:
\begin{equation*}
    \frac{p^H(h^t)-(1-\widehat{\delta})}{\widehat{\delta}} v^H
    +\frac{p^L(h^t)}{\widehat{\delta}} v^L
    +\frac{p^N(h^t)}{\widehat{\delta}} v^N.
\end{equation*}
Under these continuation values, type $c_1$ sender is indifferent between sending message $h$ and $l$ when $\omega_t=l$, and strictly prefers message $h$ when $\omega_t=h$. Types $c_2$ to $c_n$ strictly prefer to send message $h$ in both states.

\paragraph{Class 2 Active Learning History:} Play reaches a Class 2 history  if:
\begin{enumerate}
  \item $p^L(h^t) \in (0, 1-\widehat{\delta})$,
  \item The first entry of (\ref{A.5}) is non-negative.
\end{enumerate}
At those histories, the receiver plays $\mathbf{b}^T$.
All types (in the support of receiver's belief) except for type $\underline{c}(h^t)$ play $\mathbf{a}^H$ with probability $1$ while type $\underline{c}(h^t)$ (potentially) mixes between $\mathbf{a}^L$ and $\mathbf{a}^H$. To specify his mixing probabilities,
let
\begin{equation}\label{A.21}
    l(h^t) \equiv \# \Big\{
    h^s \Big| h^s\prec h^t, h^s \textrm{ belongs to Class 2}, \omega_{s}=l, \textrm{ and } \underline{c}(h^s)=\underline{c}(h^t)
    \Big\}
\end{equation}
be the number of histories that (1) strictly precede $h^t$, and (2) the lowest lying cost type in the support of the receiver's belief is
$\underline{c}(h^t)$, and (3) the state realized at history $h^s$ is $l$.
\begin{enumerate}
  \item If $\underline{c}(h^t)=c_j$ with $j \geq 3$, then
type $\underline{c}(h^t)$ plays $\mathbf{a}^L$ at $h^t$ with probability
\begin{equation}\label{A.22}
    \frac{1}{k_j-l(h^t)},
\end{equation}
and $\mathbf{a}^H$ with complementary probability, with
$k_j$ the integer defined in (\ref{A.5}).
  \item If $\underline{c}(h^t)=c_2$, then
type $c_2$ sender plays $\mathbf{a}^L$ at $h^t$ with probability
\begin{equation}\label{A.23}
    \min\{1, \frac{\rho^*}{1-\eta(h^t)}\},
\end{equation}
and $\mathbf{a}^H$ with complementary probability.
\end{enumerate}
The sender's continuation value is given by $v(h^t,(h,h))=v(h^t)$, $v(h^t,(h,l))=p^* v^H +(1-p^*) v^L$,
\begin{equation}\label{A.24}
v(h^t,(l,h)) \equiv   \frac{Q(h^t)}{\widehat{\delta}} v^H + \frac{\widehat{\delta}-Q(h^t)}{\widehat{\delta}} v^N,
  \end{equation}
  where
  \begin{equation}\label{A.25}
    Q(h^t) \equiv p^H(h^t) - (1-\widehat{\delta}) + \frac{p^L(h^t)}{p_h+(1-p_h)\underline{c}(h^t)}.
  \end{equation}
Therefore, type $\underline{c}(h^t)$ sender weakly prefers to send message $h$ when the realized state is $l$, while other types of senders in the support of receiver's belief at $h^t$ strictly prefers to send message $l$ when the realized state is $l$.
The sender's continuation value after sending message $l$ when $\omega_t=l$
   depends on whether $\eta(h^t,(l,l))$ equals $1$ or not, with
$\eta(h^t,(l,l))$ computed via Bayes Rule given the receiver's belief at $h^t$ and type $\underline{c}(h^t)$'s mixing probability at $h^t$:
\begin{enumerate}
  \item If $\eta(h^t,(l,l))<1$, then the sender's continuation payoff at $(h^t,(l,l))$ is:
  \begin{equation*}
        v(h^t,(l,l))= \frac{p^H(h^t)-(1-\widehat{\delta})}{\widehat{\delta}} v^H
    +\frac{p^L(h^t)}{\widehat{\delta}} v^L
    +\frac{p^N(h^t)}{\widehat{\delta}} v^N.
  \end{equation*}
  \item If $\eta(h^t,(l,l))=1$, then the sender's continuation payoff at $(h^t,(l,l))$ is:
  \begin{equation*}
    v(h^t,(l,l))= q(h^t) v^H +(1-q(h^t)) v^N,
\end{equation*}
where $q(h^t) \in [0,1]$ is such that the first entry of $v(h^t,(l,l))$ equals the first entry of payoff vector:
\begin{equation*}
    \frac{p^H(h^t)-(1-\widehat{\delta})}{\widehat{\delta}} v^H
    +\frac{p^L(h^t)}{\widehat{\delta}} v^L
    +\frac{p^N(h^t)}{\widehat{\delta}} v^N.
\end{equation*}
\end{enumerate}
Notice that  $\eta(h^t,(l,l))=1$ can only happen when $\underline{c}(h^t)=c_2$. This is because if $\underline{c}(h^t)<c_2$, then type $c_2$ sends message $l$ with probability $1$ when the state is $l$ at history $h^t$, which implies that the sender's reputation $\eta(h^t,(l,l))$ cannot equal $1$.

\paragraph{Rebounding Phase:} Play reaches the \textit{rebounding phase} if
\begin{enumerate}
  \item $p^L(h^t) \neq 0$,
  \item The first entry of (\ref{A.5}) is negative.
\end{enumerate}
At those histories, the receiver plays $\mathbf{b}^N$, all types of sender plays $\mathbf{a}^L$ (i.e., sending message $h$ regardless of the state), and therefore, the sender's message reveals no information about his type and the receiver's incentive constraints are satisfied.
Given the sender's continuation value at $h^t$, denoted by $v(h^t)$,
his continuation value following $(\omega_t,m_t)=(h,h)$ equals $v(h^t)$,
his continuation value after $(\omega_t,m_t)=(h,l)$ equals $v(h^t)$,
his continuation value after $(\omega_t,m_t)=(l,l)$ equals $\underline{v}$, and
his continuation value after $(\omega_t,m_t)=(l,h)$ equals
\begin{equation}\label{A.26}
v(h^t,(l,h)) \equiv \frac{p^H(h^t)}{\widehat{\delta}} v^H+\frac{p^L(h^t)}{\widehat{\delta}} v^L +\frac{p^N(h^t)-(1-\widehat{\delta})}{\widehat{\delta}} v^N.
\end{equation}
Under these continuation values, every type of sender has an incentive to send message $h$ when the state is $h$. According to (\ref{A.26}),  we have:
\begin{equation*}
    v(h^t)=(1-p_h) \Big\{
    (1-\delta) \mathbf{0} + \delta v(h^t)
    \Big\}
    +p_h \Big\{
   - (1-\delta) \mathbf{c} +\delta v(h^t,(l,l))
    \Big\},
\end{equation*}
with $\mathbf{0} \equiv (0,...,0)$ and $\mathbf{c} \equiv (c_1,...,c_n)$. Since $v(h^t) \geq 0$,
we have $v(h^t) > (1-\delta) \mathbf{0} + \delta v(h^t)$, and therefore,
\begin{equation*}
    - (1-\delta) \mathbf{c} +\delta v(h^t,(l,h)) > v(h^t) \geq \underline{v}.
\end{equation*}
This implies that the sender has an incentive to send message $h$ when the state is $l$.
\paragraph{Absorbing Phase:} Play reaches the absorbing phase
if $p^L(h^t)=0$, i.e.,
the sender's continuation value is a convex combination of $v^N$ and $v^H$, after which play stays in the absorbing phase forever and learning about the sender's type stops. By construction of the previous phases, type $c_1$'s continuation value when play first reaches the absorbing phase is non-negative, i.e., it is weakly greater than his minmax payoff $0$.

I construct, for any $p^H \geq p^*$ and under any belief of the receiver's, an equilibrium in which (1) learning does not happen on the equilibrium path, and (2) the sender attains payoff vector $v=p^Hv^H+p^N v^N$.
Recall the definition of $p^*$ in the beginning of this appendix.
If $h^t$ is such that
\begin{equation*}
    \frac{p^H(h^t)-(1-\delta)}{\delta} \geq \frac{p^*+1}{2},
\end{equation*}
then all types of senders adopt the honest strategy $\mathbf{a}^H$, and players' continuation payoffs are given by the following functions of the current period outcome:
\begin{equation}\label{A.11}
    v(h^t,(h,h))=v(h^t), \textrm{ }  v(h^t,(l,l))= \frac{p^H(h^t)-(1-\delta)}{\delta} v^H
    +\frac{p^N(h^t)}{\delta} v^N,
\end{equation}
and
\begin{equation}\label{A.12}
v(h^t,(l,h))=v(h^t,(h,l))=\underline{v}.
\end{equation}
I use $\delta$ instead of $\widehat{\delta}$ in (\ref{A.11}) since all types of senders play a pure action at such a history, so his equilibrium action can be perfectly monitored.
When $\delta$ is large enough, every type of sender strictly prefers to conform at those histories.
If $h^t$ is such that:
\begin{equation*}
   \frac{p^H(h^t)-(1-\delta)}{\delta} <\frac{p^*+1}{2},
\end{equation*}
then all types of senders play $\mathbf{a}^L$, and the sender's continuation payoff after sending message $h$ is:
\begin{equation}\label{A.14}
\frac{p^H(h^t)}{\delta} v^H
    +\frac{p^N(h^t)-(1-\delta)}{\delta} v^N,
\end{equation}
and after sending message $l$ is $p^* v^H +(1-p^*) v^L$.
One can verify that at each of these histories,
type $c_1$ sender weakly prefers to conform and types $c_2$ to $c_m$ strictly prefer to conform. Moreover, this continuation equilibrium is incentive compatible for players regardless of the receiver's belief about the sender's lying costs.

\subsection{Incentive Constraints \& Promise Keeping Constraints}
I verify players' incentive constraints and the sender's promise keeping conditions. The latter requires that the continuation play delivers all types of senders their respective continuation values. Due to the nature of incomplete information, I use a different approach compared to  Abreu, Pearce and Stacchetti (1990) by showing that
as $t \rightarrow \infty$, players' on-path play enters the absorbing phase with probability $1$. Given that in the previous subsection,
I have constructed equilibrium strategies \textit{after} play reaches the absorbing phase that are (1)
incentive compatible for all types of senders and the receivers, and (2) can deliver the promised continuation values to all types of senders,
the eventual convergence to the absorbing phase implies the promise keeping condition.

I state several results which implies that first, players' incentive constraints are satisfied at every history, and second, the promised continuation value at each history \textit{can be delivered} in the continuation equilibrium. A \textit{outcome path} in period $t$ consists of the states and the sender's messages from period $0$ to period $t-1$, with:
\begin{equation*}
    o (h^t) \equiv \Big((\omega_0,m_0),...,(\omega_{t-1},m_{t-1})\Big).
\end{equation*}
For every on-path history $h^t$, $o(h^t)$ consists only of $(h,h)$, $(l,l)$, and $(l,h)$. A \textit{reduced outcome path}, denoted by $r(h^t)$, is derived from $o(h^t)$ by ignoring periods in which the outcome is $(h,h)$. As a result, every component of $r(h^t)$
can be summarized by the sender's message alone. Abusing notation, $r(h^t) \in \bigcup_{k=0}^{\infty}\{h,l\}^k $, with the $n$th element of $r(h^t)$ be the sender's message when state $l$ occurs for the $n$th time. For every $h^t \succ h^s$, let $r(h^t\backslash h^s)$ be the reduced outcome path between $h^s$ and $h^t$.
According to the specifications of posterior beliefs and continuation values in the previous section, for every pair of histories $h^s$ and $h^t$ with $h^t \succ h^s$, suppose all histories between $h^s$ and $h^t$ belong to the active learning phase, then $\eta(h^t)$ only depends on $\eta(h^s)$ and the sequence $r(h^t\backslash h^s)$.
\begin{itemize}
  \item Let
$N_l(r(h^t\backslash h^s))$ be the number of message $l$ in
$r(h^t\backslash h^s)$.
  \item Let $N_h(r(h^t\backslash h^s))$ be the number of message $h$ in
$r(h^t\backslash h^s)$.
  \item Let $|r(h^t\backslash h^s)|$ be the number of elements in $r(h^t\backslash h^s)$.
  \item Let
  \begin{equation*}
    r(h^t\backslash h^s) \equiv \Big(m_0,m_1,...,m_{|r(h^t\backslash h^s)|-1}\Big).
  \end{equation*}
\end{itemize}
For every $h^t \succ h^s$ with all histories from $h^s$ to $h^t$ belonging to Class 1,
the receiver's posterior belief attaches the following probability to type $c_1$ sender
\begin{equation}\label{A.27}
 \eta(h^t)=   \eta^*+ (\eta(h^s)-\eta^*) \Big(1-\lambda(1-\rho^*) \Big)^{N_h(r(h^t\backslash h^s))}
    \Big( 1+\lambda \rho^* \Big)^{N_l(r(h^t\backslash h^s))},
\end{equation}
and the convex weights on the sender's continuation value are given by:
\begin{equation*}
    p^H(h^t)= \frac{\displaystyle p^H(h^s)-(1-\widehat{\delta})\sum_{\tau=1}^{|r(h^t\backslash h^s)|} \widehat{\delta}^{\tau}\mathbf{1}\{m_{\tau}=l\}}{\displaystyle \widehat{\delta}^{|r(h^t\backslash h^s)|}},
\end{equation*}
\begin{equation*}
    p^L(h^t)= \frac{\displaystyle p^L(h^s)-(1-\widehat{\delta})\sum_{\tau=1}^{|r(h^t\backslash h^s)|} \widehat{\delta}^{\tau}\mathbf{1}\{m_{\tau}=h\}}{\displaystyle \widehat{\delta}^{|r(h^t\backslash h^s)|}},
\end{equation*}
\begin{equation*}
    p^N(h^t)= \frac{\displaystyle p^N(h^s)}{\displaystyle \widehat{\delta}^{|r(h^t\backslash h^s)|}}.
\end{equation*}
The following lemma shows that conditional play remains in Class 1 histories, the discounted average frequency of message $l$ divided by the discounted average frequency of message $h$ in the reduced outcome  (i.e., counting only periods where the state realization is $l$) must be below some cutoff. This provides an upper bound on the sender's continuation value, which is directly implied by Lemma A.1 in Pei (2019):
\begin{Lemma1}
For every $\underline{\eta} \in (\eta^*,1)$,
there exist $T \in \mathbb{N}$ and $\underline{\delta} \in (0,1)$,
s.t. when $\eta (h^s) \geq \underline{\eta}$ and $\widehat{\delta}>\underline{\delta}$,
if $h^t \equiv (y_0,...,y_{t-1}) \succ h^s$ and all histories between $h^s$ and $h^t$ belong to Class 1, then:
\begin{equation}\label{A.28}
    (1-\widehat{\delta}) \sum_{\tau=1}^{|r(h^t\backslash h^s)|} \widehat{\delta}^{\tau-1} \mathbf{1}\{m_{\tau}=l\}
    \leq
    (1-\widehat{\delta}^T)
    +
    (1-\widehat{\delta})\sum_{\tau=1}^{|r(h^t\backslash h^s)|} \widehat{\delta}^{\tau-1} \mathbf{1}\{m_{\tau}=h\}
    \cdot
    \frac{1-\widetilde{\rho}}{\widetilde{\rho}}.
\end{equation}
\end{Lemma1}
The proof is in Pei (2019). I use this result and my construction in the previous section to establish the following four lemmas, which imply incentive compatibility and promise keeping, and are shown in the follow-up subsections.
Lemma \ref{L1} establishes a lower bound on the receiver's posterior belief after observing message $l$ in state $l$ at any history that belongs to Class 2.
\begin{Lemma}\label{L1}
For any history $h^t$ belonging to Class 2,
\begin{itemize}
  \item If $\underline{c}(h^t) \leq c_3$, then $\eta(h^t,(l,l)) \geq \eta(h^0)$  and $\eta(h^t,(l,h)) =0$.
  \item If $\underline{c}(h^t)=c_2$, then $\eta(h^t,(l,l)) = \min \{1,\frac{\eta(h^t)}{1-\rho^*}\}$ and $\eta(h^t,(l,h)) =0$.
\end{itemize}
\end{Lemma}
The next lemma establishes a uniform upper bound on the number histories that belong to Class 2 and the realized state in the previous period is $l$, which applies along every on-path play.
\begin{Lemma}\label{L2}
There exist $\underline{\delta} \in (0,1)$ and $M \in \mathbb{N}$,
  such that when $\delta >\underline{\delta}$ and along every on-path play, the number of histories that belong to Class 2 while the previous period state is $l$ is $\leq M$.
\end{Lemma}
Lemma \ref{L2} implies that for every on-path history $h^t$, the number of periods that belong to Class 2 in the reduced outcome $r(h^t)$
is no more than $M$.
Lemma \ref{L3} establishes a uniform lower bound on $p^H(h^t)$ for all histories belonging to the active learning phase (Class 1 and Class 2).
\begin{Lemma}\label{L3}
There exist $\underline{\delta} \in (0,1)$ and $\underline{Q}>0$, such that when $\delta > \underline{\delta}$, we have $p^H(h^t) \geq \underline{Q}$ for all $h^t$ belonging to the active learning phase.
\end{Lemma}
Lemma \ref{L3} also implies a lower bound on $p^H(h^t)$ if $h^t$ is the \textit{first history} that reaches the absorbing phase, i.e., $h^t$ is such that $p^L(h^t)=0$ and $p^L(h^s)>0$ for all $h^s \prec h^t$.
\begin{Lemma}
There exist $\underline{\delta} \in (0,1)$ and $M \in \mathbb{N}$,
  such that when $\delta >\underline{\delta}$ and along every on-path play, the number of histories that belong to the rebounding phase is at most $M$.
\end{Lemma}

\subsection{Proof of Lemma A.1}
\paragraph{Case 1:} Consider the case in which  $\underline{c}(h^t) \leq c_3$. First, suppose $\eta(h^t) \geq \eta(h^0)$, then
the conclusion of
Lemma A.1 follows since
$\eta(h^t,(l,l))>\eta(h^t) \geq \eta(h^0)$.
Second, suppose $\eta(h^t)< \eta(h^0)$, then given the value of $l(h^t)$
and the lowest-cost type at $h^t$ being $c_j$,
the posterior probability of type $c_1$ is bounded from below by:
\begin{eqnarray*}
 & &   \frac{\eta(h^t)}{\displaystyle \eta(h^t) +(1-\eta(h^t)) \frac{\pi_2+...+\pi_{j-1} + \frac{k_j-l(h^t)-1}{k_j} \pi_j}{\pi_2+...+\pi_{j-1} + \frac{k_j-l(h^t)}{k_j} \pi_j}}
\geq
\frac{\eta(h^t)}{\displaystyle \eta(h^t) +(1-\eta(h^t)) \frac{\pi_2+...+\pi_{j-1} + \frac{k_j-1}{k_j} \pi_j}{\pi_2+...+\pi_{j-1} +\pi_j}}
\end{eqnarray*}
Let
\begin{equation*}
X \equiv 1- \frac{\pi_2+...+\pi_{j-1} + \frac{k_j-1}{k_j} \pi_j}{\pi_2+...+\pi_{j-1} +\pi_j}=
\frac{\pi_j}{k_j (\pi_2+...+\pi_{j-1} +\pi_j)}.
\end{equation*}
The lower bound on posterior belief $\frac{\eta(h^t)}{\eta(h^t)+(1-\eta(h^t)) (1-X)}$ is greater than $\pi_1$ if and only if:
\begin{equation*}
    X \geq 1-\frac{(1-\pi_1)\eta(h^t)}{\pi_1 (1-\eta(h^t))}
    =\frac{\pi_1-\eta(h^t)}{\pi_1 (1-\eta(h^t))}.
\end{equation*}
Given that $\eta(h^t)\geq \eta^*$ at every history $h^t$ that belongs to Class 2,  the above inequality is implied by (\ref{A.17}).

\paragraph{Case 2:}
Consider the case in which $\underline{c}(h^t)=c_2$. If $\eta(h^t) \geq 1-\rho^*$, then type $c_2$
plays $\mathbf{a}^L$ with probability $\min \{1,\frac{\rho^*}{1-\eta(h^t)}\}=1$, which implies that $\eta(h^t,(l,l))=1$.
If $\eta(h^t) < 1-\rho^*$, then type $c_2$
 lies in state $l$ with probability $\min \{1,\frac{\rho^*}{1-\eta(h^t)}\}=\frac{\rho^*}{1-\eta(h^t)}$, which implies that $\eta(h^t,(l,l))=\eta(h^t)/(1-\rho^*) \geq (1-\rho^*) \eta(h^0)/(1-\rho^*)=\eta(h^0)$.

\subsection{Proof of Lemma A.2}
\paragraph{Step 1:} If $h^t$ belongs to Class 2 and $\underline{c}(h^t) =c_j \leq c_3$, then
type $\underline{c}(h^t)$ sends message $h$ in state $l$ with probability $1$
when $l(h^t)=k_j-1$, after which play reaches the absorbing phase. Therefore, along every path of play, there are at most $k_j$ Class 2 histories satisfying $\underline{c}(h^t)=c_j$ and the state in the previous period (i.e. period $t-1$) is $l$.
This further implies that
there are at most
\begin{equation*}
K \equiv k_3+...+k_n
\end{equation*}
Class 2 histories that has $\underline{c}(h^t) \leq c_3$ and the previous period state being $l$.

\paragraph{Step 2:} Consider the number of Class 2 histories such that (1)
$\underline{c}(h^t)=c_2$, and (2) the state in the previous period is $l$.
Let $N \equiv \lceil \frac{1}{1-\gamma}\rceil$, and recall the integer constant $T$ in Lemma A.1. In addition to the requirements on $\delta$ specified before, I need require $\delta$ to be large enough such that $\widehat{\delta}$ satisfies
\begin{equation}\label{A.29}
    \widehat{\delta}^{T+1} (1+\widehat{\delta}+...+\widehat{\delta}^N) > N \textrm{ and } 2 \widehat{\delta}^{T+N+2} > 1.
\end{equation}
First, I show that after the sender sends message $l$ when the state is $l$ at $h^t$, it takes at most $T+N$ such periods for play to reach a history that belongs to either to the absorbing phase or to another Class 2 history.
According to the continuation value at $(h^t,(l,l))$, we have:
\begin{equation}\label{A.30}
p^L(h^t,(l,l))=\frac{p^L(h^t)}{\widehat{\delta}} < \frac{1-\widehat{\delta}}{\widehat{\delta}}.
\end{equation}
The last inequality comes from $h^t$ belonging to Class 2, so that $p^L(h^t)<1-\widehat{\delta}$ by definition.
According to Lemma A.1, for every Class 1 history $h^s$ such that $h^s \succ (h^t,(l,l)) \equiv h^{t+1}$ and all histories between $(h^t,(l,l))$ and $h^s$ belong to Class 1, we have:
\begin{equation}\label{A.31}
  (1-\widehat{\delta}) \sum_{\tau=1}^{|r(h^s\backslash h^{t+1})|} \widehat{\delta}^{\tau-1} \mathbf{1}\{m_{\tau}=l\}
    \leq
    (1-\widehat{\delta}^T)
    +
    (1-\widehat{\delta})\sum_{\tau=1}^{|r(h^s\backslash h^{t+1})|} \widehat{\delta}^{\tau-1} \mathbf{1}\{m_{\tau}=h\}
    \cdot
    \frac{1-\widetilde{\rho}}{\widetilde{\rho}}.
\end{equation}
Moreover, (\ref{A.30}) and the requirement that all histories between $(h^t,(l,l))$ and $h^s$ belong to Class 1 imply that
\begin{equation}\label{A.32}
   (1-\widehat{\delta})\sum_{\tau=1}^{|r(h^s\backslash h^{t+1})|} \widehat{\delta}^{\tau-1} \mathbf{1}\{m_{\tau}=h\}
    < \frac{1-\widehat{\delta}}{\widehat{\delta}}.
\end{equation}
Given that only $(\mathbf{a}^L,\mathbf{b}^T)$ and $(\mathbf{a}^H,\mathbf{b}^T)$ occur at active learning phase histories (Class 1 and 2):
\begin{equation*}
    1-\widehat{\delta}^{|r(h^s\backslash h^{t+1})|-1}
    =   (1-\widehat{\delta}) \sum_{\tau=1}^{|r(h^s\backslash h^{t+1})|} \widehat{\delta}^{\tau-1} \mathbf{1}\{m_{\tau}=h\}
    + (1-\widehat{\delta}) \sum_{\tau=1}^{|r(h^s\backslash h^{t+1})|} \widehat{\delta}^{\tau-1} \mathbf{1}\{m_{\tau}=l\}
\end{equation*}
\begin{equation}\label{A.33}
    \leq (1-\widehat{\delta}^T)+ \frac{1-\widehat{\delta}}{\widehat{\delta}}+ \frac{1-\widehat{\delta}}{\widehat{\delta}} \frac{1-\widetilde{\rho}}{\widetilde{\rho}}
    \leq (1-\widehat{\delta}^T) +\frac{1-\widehat{\delta}}{\widehat{\delta} \widetilde{\rho}} \leq (1-\widehat{\delta}^T) +\frac{1-\widehat{\delta}}{\widehat{\delta}\rho}
\end{equation}
Next, I show that $|r(h^s\backslash h^{t+1})| \leq T+N+1$. Suppose toward a contradiction that $|r(h^s\backslash h^{t+1})| \geq T+N+2$, then
\begin{equation*}
 (1-\widehat{\delta}^T) +\frac{1-\widehat{\delta}}{\widehat{\delta}}N \geq (1-\widehat{\delta}^T) +\frac{1-\widehat{\delta}}{\widehat{\delta}\rho} \geq  1-\widehat{\delta}^{|r(h^s\backslash h^{t+1})|-1} \geq 1-\widehat{\delta}^{T+N+1},
\end{equation*}
which yields:
 \begin{equation*}
\frac{1-\widehat{\delta}}{\widehat{\delta}}N \geq  \widehat{\delta}^T(1-  \widehat{\delta}^{N+1}).
 \end{equation*}
Dividing both sides by $\frac{1-\widehat{\delta}}{\widehat{\delta}}$, we have:
\begin{equation*}
 N \geq    \widehat{\delta}^{T+1} (1+\widehat{\delta}+...+\widehat{\delta}^N),
\end{equation*}
which contradicts the first inequality of (\ref{A.29}).

Second, I focus on history $h^s$ that has the following two features:
\begin{itemize}
  \item[1.] $h^s$ belongs to Class 2, and the state in period $s-1$ is $l$,
  \item[2.] $h^s \succeq (h^t,(l,l))$ and
  all histories between $(h^t,(l,l))$ and $h^s$, excluding $h^s$, belong to Class 1.
\end{itemize}
I show that there exists at most one period from $(h^t,(l,l))$ to $h^s$ such that the stage-game outcome is such that the sender sends message $h$ while the state is $l$. Suppose toward a contradiction that there exist two or more such periods, then
\begin{equation*}
   (1-\widehat{\delta})\sum_{\tau=1}^{|r(h^s\backslash h^{t+1})|} \widehat{\delta}^{\tau-1} \mathbf{1}\{m_{\tau}=h\} \geq 2(1-\widehat{\delta})\widehat{\delta}^{T+N+1}.
\end{equation*}
The last inequality comes from the previous conclusion that $|r(h^s\backslash h^{t+1})| \leq T+N+1$.
According to (\ref{A.32}),
\begin{equation}\label{A.34}
2(1-\widehat{\delta})\widehat{\delta}^{T+N+1}  <  (1-\widehat{\delta})\sum_{\tau=1}^{|r(h^s\backslash h^{t+1})|} \widehat{\delta}^{\tau-1} \mathbf{1}\{m_{\tau}=h\}< \frac{1-\widehat{\delta}}{\widehat{\delta}}.
\end{equation}
The above inequality contradicts the second inequality of (\ref{A.29}) that $2 \widehat{\delta}^{T+N+2} > 1$.

Let $h^t$ be the first time play reaches a history that belongs to Class 2 with $\underline{c}(h^t)=c_2$. According to Lemma A.1, $\eta(h^t,(l,l)) \geq \frac{\eta^*}{1-\rho^*} \geq \eta(h^0)=\pi_1$.
Let $h^s$ be the next history that belongs to Class 2 with $\omega_{s-1}=l$. Since we have shown that the receiver takes the wrong action
at most once between $(h^t,(l,l))$ and $h^s$, we know that
\begin{equation*}
    \eta(h^s, (l,l) =\min\{1, \frac{\eta(h^s)}{1-\rho^*} \}
    \geq \min \{1, \frac{\eta(h^t,(l,l))}{1-\rho^*} (1-\lambda (1-\rho^*))\}
\end{equation*}
Therefore, conditional on $(h^s, (l,l))$ is not an absorbing phase history, the receiver's belief at $(h^s, (l,l))$ attaches
probability at least:
\begin{equation}\label{A.35}
 \eta(h^s, (l,l)) \geq   \eta(h^t,(l,l)) \frac{1-\lambda (1-\rho^*)}{1-\rho^*} \geq  \eta(h^t,(l,l)) \sqrt{\frac{1}{1-\rho^*}}
\end{equation}
to type $c_1$,
where the last inequality comes from $\lambda \in (0,\frac{1-\sqrt{1-\rho^*}}{1-\rho^*})$.
Let
\begin{equation*}
    \widehat{M} \equiv \frac{\log (1/\pi_1)}{\log \sqrt{\frac{1}{1-\rho^*}}}+1.
\end{equation*}
Since $\eta(h^t,(l,l)) \geq \pi_1$ for the first Class 2 history $h^t$ satisfying $\underline{c}(h^t)=c_2$, there can be at most $\widehat{M}$ Class 2 histories
with $c_2$ being the highest-cost type along every path of play. This is because otherwise,
the receiver's posterior belief attaches probability greater than
\begin{equation*}
    \pi_1 \Big(\frac{1}{\sqrt{1-\rho^*}}\Big)^{\widehat{M}} >1
\end{equation*}
at the $(\widehat{M}+1)$th such history, which leads to a contradiction.
Summarizing the conclusions of the two parts, we know that
along every path of equilibrium play,
there exist at most $M \equiv K+\widehat{M}$  histories  that belong to Class 2 and the state in the previous period is $l$.

\subsection{Proof of Lemma A.3}
Consider any given Class 2 history $h^t$ such that no predecessor of $h^t$ belongs to Class 2, in another word, all predecessors of $h^t$ belong to Class 1 or the rebounding phase. Therefore, $p^H(h^{t-1}) \geq Y$, which implies that $p^H(h^{t}) \geq Y-(1-\widehat{\delta})$. As a result
\begin{equation*}
    Q(h^t)=p^H(h^t) - \frac{1-\widehat{\delta}-p^L(h^t)}{\underline{c}(h^t)} \geq Y-(1-\widehat{\delta}) >0.
\end{equation*}
If play remains at the active learning phase (Class 1 or 2) after $h^t$, then player $1$ must be sending message $l$ when the state is $l$ at $h^t$, after which
\begin{equation*}
    p^H(h^t,(l,l)) \geq p^H(h^t)-(1-\widehat{\delta}) \geq Y-2(1-\widehat{\delta}) \textrm{ and } p^L(h^t,(l,l)) \leq \frac{1-\widehat{\delta}}{\widehat{\delta}}.
\end{equation*}
According to Lemma A.1, $\eta(h^t,(l,l)) \geq \eta(h^0)=\pi_1$. One can then apply Lemma A.0 again, which implies that at every Class 1 history $h^s$ such that only one predecessor of $h^s$ (1) belongs to Class 2 and (2) $\omega_{s-1}=l$, we have:
\begin{equation*}
    p^H(h^s) \geq Z \equiv Y-2(1-\widehat{\delta})-\frac{1-\widehat{\delta}}{\widehat{\delta}}\frac{1-\widetilde{\rho}}{\widetilde{\rho}}-(1-\widehat{\delta}^T).
\end{equation*}
 When $\widehat{\delta}$ is large enough, $Z \geq Y/2$. Therefore, for every Class 2 history $h^s$ such that there is only one strict predecessor history belongs to Class 2,
\begin{equation*}
        Q(h^s)=p^H(h^s) - \frac{1-\widehat{\delta}-p^L(h^s)}{\underline{c}(h^s)} \geq Z-(1-\widehat{\delta}) >0.
\end{equation*}
Iterative this process. Since
\begin{enumerate}
  \item the number of Class 2 histories with the state in the previous period being $l$  is bounded from above by $M$ along every path of play,
  \item for every Class 2 history $h^t$, $p^L(h^t,(l,l)) =\frac{1-\widehat{\delta}}{\widehat{\delta}}$ and
 $\eta(h^t,(l,l)) \geq \eta(h^0)$,
\end{enumerate}
there exist $\underline{\delta} \in (0,1)$ and $\underline{Q}>0$ such that when $\delta >\underline{\delta}$, $p^H(h^t) \geq \underline{Q}$ for every Class 1 history $h^t$.
\subsection{Proof of Lemma A.4}
 I construct a constant $K \in \mathbb{N}$ that is independent of the discount factor $\delta$ such that once play enters the rebounding phase,
it will go back to the active learning phase after \textit{at most} $K$ periods with the realized state $\omega$ being $l$.
First, type $c_1$'s continuation value in the rebounding phase is at least $0$. Second, play goes back to the active learning phase whenever his continuation value is above $(1-\widehat{\delta}) (p_h+(1-c_1)(1-p_h))$. After $K$ periods in the rebounding phase with the realized state being $l$, type $c_1$'s continuation value is at least:
\begin{equation}\label{A.33}
    \frac{1-\widehat{\delta}^K}{\widehat{\delta}^K} c(1-p_h),
\end{equation}
which is more than $(1-\widehat{\delta}) (p_h+(1-c_1)(1-p_h))$ if
\begin{equation}\label{A.34}
    K \geq \Big\lceil
    \frac{p_h+(1-c_1)p_h}{c_1 (1-p_h)^2}
    \Big\rceil.
\end{equation}
Lemma A.4 is obtained by setting $\delta$ to be close enough to $1$ such that the corresponding $\widehat{\delta}^K$ satisfies the requirement for $\delta$ in Lemma A.0.
\section{Proof of Theorem 3: Statement 1}\label{secB}
I show Lemma \ref{L5.1} in section \ref{subB.1}, and show Lemma \ref{L5.2} in section \ref{subB.2}.
The two lemmas together imply the first statement of Theorem \ref{Theorem3}.
\subsection{Proof of Lemma 5.1}\label{subB.1}
\paragraph{Step 1:} Given that $\overline{v}(c)=p_h(2-c)$, the optimal value is attained by a distribution that attaches probability $\rho^*$ to action profile $(\mathbf{a}^L,\mathbf{b}^T)$ and probability $1-\rho^*$ to action profile $(\mathbf{a}^H,\mathbf{b}^T)$. I denote this distribution by $\gamma^* \in \Delta (\mathbf{A} \times \mathbf{B})$, under which type $c_j$'s expected stage-game payoff is $p_h(2-c_j)$ for every $c_j \in \mathcal{C} \cap [0,1)$. Therefore, $\overline{v}(c)$ is also the value of the following constrained optimization problem, which shares the same objective function but faces a larger set of constraints:
\begin{equation*}
    \overline{v} (c) = \max_{\gamma \in \Delta (\mathbf{A} \times \mathbf{B})} \sum_{(\mathbf{a},\mathbf{b}) \in \mathbf{A} \times \mathbf{B}} \gamma (\mathbf{a},\mathbf{b}) u_s(c,\mathbf{a},\mathbf{b})
\end{equation*}
subject to the constraint that for \textit{every} $c_j \in \mathcal{C} \cap [0,1)$,
\begin{equation*}
\sum_{(\mathbf{a},\mathbf{b}) \in \mathbf{A} \times \mathbf{B}} \gamma (\mathbf{a},\mathbf{b}) u_s(c_j,\mathbf{a},\mathbf{b})
\geq p_h (2-c_j).
\end{equation*}
\paragraph{Step 2:} Fix the value of $\delta$.
Let $\sigma_{c_j}: \mathcal{H} \rightarrow \Delta (\mathbf{A})$ be type $c_j$ sender's equilibrium strategy, and let $\sigma_r: \mathcal{H} \rightarrow \Delta (\mathbf{B})$ be the receiver's equilibrium strategy. Let $\gamma^c \in \Delta (\mathbf{A} \times \mathbf{B})$ be defined as:
\begin{equation*}
     \gamma^c(\mathbf{a},\mathbf{b}) \equiv  \mathbb{E}^{(\sigma_{c},\sigma_r)} \Big[
    \sum_{t=0}^{\infty} (1-\delta)\delta^t \mathbf{1} \{
(\mathbf{a}_t,\mathbf{b}_t) =   (\mathbf{a},\mathbf{b})
     \}
    \Big], \textrm{ for every } (\mathbf{a},\mathbf{b}) \in \mathbf{A} \times \mathbf{B}.
\end{equation*}
Let $\Sigma_{c_j}$ be the set of pure strategies in the support of $\sigma_{c_j}$, with $\widehat{\sigma}_{c_j} : \mathcal{H} \rightarrow \mathbf{A}$ a typical element. Let
\begin{equation}\label{B.1}
 d(c,c_j) \equiv   \sum_{(\mathbf{a},\mathbf{b}) \in \mathbf{A} \times \mathbf{B}} \gamma^c (\mathbf{a},\mathbf{b}) u_s(c,\mathbf{a},\mathbf{b})
    -\sup_{\widehat{\sigma}_{c_j} \in \Sigma_{c_j}}\mathbb{E}^{(\widehat{\sigma}_{c_j},\sigma_2)} \Big[
    \sum_{t=0}^{\infty} (1-\delta)\delta^t u_s(c,\mathbf{a},\mathbf{b})
    \Big].
\end{equation}
Let $c_j \equiv \max \Big\{ \mathcal{C} \cap [0,1) \Big\}$.
I derive an upper bound for
\begin{equation}\label{B.2}
  \eta(c_j) \equiv \min_{c \in \mathcal{C} \cap [1,+\infty) } d(c,c_j).
\end{equation}
For every $c \in  \mathcal{C} \cap [1,+\infty)$ and every
pure strategy $\widehat{\sigma}_{c_j}$ in the support of $\sigma_{c_j}$,
$\widehat{\sigma}_{c_j}$ is not
an
$\eta(c_j)$-best reply against $\sigma_r$ for type $c$ sender, which implies that for all on-path histories after period $T \in \mathbb{N}$, type $c_j$ sender is separated from all types who have strictly higher lying costs, with $T$ the largest integer satisfying:
\begin{equation}\label{B.3}
    (1+c_1) \delta^T \geq \eta(c_j),
\end{equation}
where $1+c_1$ is the largest difference in stage-game payoff for any type of sender. According to Proposition \ref{Prop1}, type $c_j$'s continuation value after period $T$ is at most $p_h$. Given that his equilibrium payoff is at least $v_j^{**}-\varepsilon$, we have:
\begin{equation*}
\mathbb{E}^{(\sigma_{c_j},\sigma_2)} \Big[ \sum_{t=0}^T (1-\delta)\delta^t u_1(c_j,\mathbf{a},\mathbf{b})  \Big]
+\delta^T p_h
\geq
 v_j^{**}-\varepsilon,
\end{equation*}
or equivalently,
\begin{equation}\label{B.4}
\mathbb{E}^{(\sigma_{c_j},\sigma_2)} \Big[ \sum_{t=0}^T (1-\delta)\delta^t u_1(c_j,\mathbf{a},\mathbf{b})  \Big] \geq    v_j^{**}-\varepsilon -\delta^T p_h.
\end{equation}
According to Gossner (2011), there exists $S \in \mathbb{N}$ that depends only on $\pi(c_j)$ such that:
\begin{equation}\label{B.5}
    \mathbb{E}^{(\sigma_{c_j},\sigma_2)} \Big[ \sum_{t=0}^T (1-\delta)\delta^t u_1(c_j,\mathbf{a},\mathbf{b})  \Big]
    \leq (1-\delta^T) v_j^{**} + (1-\delta^{S}) (1-v_j^{**})
\end{equation}
This together with (\ref{B.4}) implies that:
\begin{equation*}
    (1-\delta^T) v_j^{**} + (1-\delta^{S}) (1-v_j^{**}) \geq v_j^{**}-\varepsilon -\delta^T p_h,
\end{equation*}
or equivalently,
\begin{equation}\label{B.6}
    \delta^T \leq \frac{\varepsilon +(1-\delta^S)(1-v_j^{**})}{v_j^{**}-p_h}.
\end{equation}
This together with (\ref{B.3}) implies the following upper bound on $\eta(c_j)$:
\begin{equation*}
    \eta(c_j) \leq \frac{1+c_1}{v_j^{**}-p_h} \big(\varepsilon+ (1-\delta^S) (1-v_j^{**})\big),
\end{equation*}
which vanishes to $0$ as $\varepsilon \rightarrow 0$ and $\delta \rightarrow 1$.
\paragraph{Step 3:} Step 1 and Step 2 imply that in every equilibrium in which some non-ethical type $c_j$ attains payoff more than $v_j^{**}-\varepsilon$, there exists an ethical type $c$ and a pure strategy in the support of $\sigma_{c_j}$ such that this pure strategy is type $c$ sender's $\eta(c_j)$-best reply, with $\eta(c_j)$ vanishes as $\varepsilon \rightarrow 0$ and $\delta \rightarrow 1$. Let $\gamma \in \Delta (\mathbf{A} \times \mathbf{B})$ be the occupation measure induced by this pure strategy. Constraint (\ref{5.2}) is necessary given that type $c_j$ obtains his equilibrium payoff from this pure strategy which is greater than $v_j^{**}-\varepsilon$, and the definition of $\overline{v}(c)$ in (\ref{5.1}) suggests that type $c$ cannot obtain payoff strictly more than $\overline{v}(c)+\eta(c_j)$ in equilibrium. This establishes Lemma \ref{L5.1}.
\subsection{Proof of Lemma 5.2}\label{subB.2}
Recall that $\gamma^1 \in \Delta (\mathbf{A} \times \mathbf{B})$ is defined as:
\begin{equation*}
     \gamma^1(\mathbf{a},\mathbf{b}) \equiv  \mathbb{E}^{(\sigma_{c_1},\sigma_r)} \Big[
    \sum_{t=0}^{\infty} (1-\delta)\delta^t \mathbf{1} \{
(\mathbf{a}_t,\mathbf{b}_t) =   (\mathbf{a},\mathbf{b})
     \}
    \Big], \textrm{ for every } (\mathbf{a},\mathbf{b}) \in \mathbf{A} \times \mathbf{B}.
\end{equation*}
Propositions \ref{Prop2} and \ref{Prop3} apply to settings with ethical types and therefore, $\gamma^1$ satisfies constraint (\ref{5.5}) in the $\delta \rightarrow 1$ limit. For type $c_1$, he can guarantee stage-game payoff $0$ by telling the truth in both states, and therefore, $\gamma^1$ satisfies constraint (\ref{5.4}). Therefore, $\underline{v}(c)$ is type $c$'s lowest possible payoff by imitating the equilibrium strategy of type $c_1$. Since type $c$'s equilibrium payoff is weakly higher, his equilibrium payoff is no less than $\underline{v}(c)$.

\section{Proof of Theorem 3: Statement 2}\label{secC}
Let $c^* \equiv \min \{\mathcal{C} \cap [1,+\infty)\}$. Let $\overline{\mathcal{C}}$ be the set of ethical types with lying costs \textit{strictly} greater than $c^*$, which can be empty. Let $\underline{\mathcal{C}} \equiv \mathcal{C} \cap [0,1)$.
Recall the definitions of $v^H$, $v^L$, and $v^N$. The second statement of Theorem \ref{Theorem3} is implied by Proposition \ref{Prop5}:
\begin{Proposition}\label{Prop5}
Suppose $c_1(c^*-1) \leq 2$.
For every $\varepsilon>0$ and $\rho \in [0,\rho^*)$, there exists $\underline{\delta} \in (0,1)$ such that for every
$\delta>\underline{\delta}$ and
$\pi$ with $\pi(c^*) \geq \varepsilon$, there exists a sequential equilibrium in which
\begin{itemize}
  \item type $c_j$ attains payoff $\rho v_j^L +(1-\rho) v_j^H$
for every $c_j \in \underline{\mathcal{C}} \cup \{c^*\}$.
\end{itemize}
\end{Proposition}
The proof consists of two steps.
Let us start from considering an \textit{auxiliary game} in which all types in $\overline{\mathcal{C}}$ occur with probability $0$. Using a similar construction as that in Appendix \ref{secA}, one can obtain an equilibrium that attains payoff vector $\rho v^L + (1-\rho) v^H$ for all types except for the ones in
$\overline{\mathcal{C}}$. I omit this construction to avoid repetition.

Let $\sigma_{c}$ be the equilibrium strategy of type $c \in \mathcal{C} \backslash \overline{\mathcal{C}}$ and let $\sigma_r$ be the receiver's equilibrium strategy. Let $\mathcal{H}^c$ be the set of histories that occur with positive probability under $(\sigma_c,\sigma_r)$. Let
$\mathcal{H}^*$ be the set of histories that
\begin{enumerate}
  \item occur with positive probability under $(\sigma_{c^*},\sigma_r)$
  \item the sender has lied at least once before the receivers' posterior belief ruling out the possibility of all types in $\underline{\mathcal{C}}$.
\end{enumerate}
Let $\mathcal{H}^{**} \equiv \Big\{ \bigcup_{c \in \underline{\mathcal{C}}} \mathcal{H}^c \Big\} \bigcup \mathcal{H}^*$.

Next, I specify the equilibrium strategies of types in $\overline{\mathcal{C}}$. I also modify $\sigma_r$ and $\{\sigma_c\}_{c \in \mathcal{C} \backslash \overline{\mathcal{C}}}$ at histories that do not belong to $\mathcal{H}^{**}$.
Every type in $\overline{\mathcal{C}}$ plays $\mathbf{a}^H$ with probability $1$ when the receiver's posterior belief attaches positive probability to types in $\underline{\mathcal{C}}$. Upon reaching $h^t$ which is the first history such that the receivers' posterior belief rules out all types in $\underline{\mathcal{C}}$, then for every type $c \in \underline{\mathcal{C}} \cup \{c^*\}$,
\begin{itemize}
  \item if $c_1 (c-1) \leq 2$, then type $c$ sends message $h$ at $h^t$ regardless of $\omega_t$,
  \item if $c_1 (c-1)>2$, then type $c$ sends message $l$ at $h^t$ regardless of $\omega_t$,
  \item the receiver plays action $L$ at $h^t$ regardless of the sender's message.
\end{itemize}
At history $h^{t+1} \succ h^t$, if message $h$ is sent in period $t$, then the continuation equilibrium consists only of outcomes $(\mathbf{a}^H,\mathbf{b}^T)$ and $(\mathbf{a}^L,\mathbf{b}^N)$, with the discounted average frequency of outcome $(\mathbf{a}^H,\mathbf{b}^T)$ equals $\widetilde{\rho}(h^t)$, which can be computed via:
\begin{equation}\label{C.1}
    v_{c^*}(h^t)= - (1-\delta)(1-p_h) c^*
    +\delta \Big(
    \widetilde{\rho}(h^t) p_h -(1-\widetilde{\rho}(h^t)) (1-p_h)c^*
    \Big),
\end{equation}
where $v_{c^*}(h^t)$ is type $c^*$ sender's continuation payoff at $h^t$.

At history $h^{t+1} \succ h^t$, if message $l$ is sent in period $t$, then the continuation equilibrium
enters a \textit{punishment phase} which is constructed in Appendix C.1 according to the proof of Proposition \ref{Prop6}. This continuation equilibrium delivers  payoff approximately $\underline{v}(c)$ to type $c$.

The receivers' incentives at every history in $\mathcal{H}^{**}$ remain intact since types in $\overline{\mathcal{C}}$ occur with zero probability at those histories. Similarly, at histories where the sender has never lied before, and type $c^*$ and types in $\underline{\mathcal{C}}$ coexist, all types in
$\overline{\mathcal{C}}$ play $\mathbf{a}^H$ with probability $1$, which strengthens the receivers' incentives to play $\mathbf{b}^T$.
For every $c \in \overline{\mathcal{C}}$, type $c$ sender's incentive to play $\mathbf{a}^H$ at the active learning phase histories follows from a supermodularity argument, in particular, among all pure strategies in the support of $\sigma_{c^*}$, playing
$\mathbf{a}^H$ at the active learning phase history minimizes the discounted average frequency of lying. Given type $c^*$ sender's indifference between these pure strategies, type $c$ sender strictly prefers pure strategies that prescribe $\mathbf{a}^H$ at those histories.

\subsection{Constructing Punishment Phase Strategy}\label{subC.1}
Recall the definitions of $v^N$ and $v^H$ in section 4.
Let $v^O \equiv (0,...,0) \in \mathbb{R}^n$ be the sender's payoff from action profile $(\mathbf{a}^H,\mathbf{b}^N)$.
For every $\rho' \in [ \rho^*, 1]$, let
\begin{equation}\label{C.2}
    w(\rho') \equiv \frac{\rho' (1-p_h) c_1}{p_h +\rho' (1-p_h) c_1}v^H
    + \frac{(1-\rho') p_h}{p_h +\rho' (1-p_h) c_1}v^O
    + \frac{\rho' p_h}{p_h +\rho' (1-p_h) c_1}v^N \in \mathbb{R}^n.
\end{equation}
One can verify that the first entry of $w(\cdot)$ equals $0$, and the other entries are strictly increasing in $\rho$. According to (\ref{5.7}),
\begin{equation}\label{C.3}
w(\rho^*)= \Big(
\underline{v}(c_1),...,\underline{v}(c_n)
\Big).
\end{equation}
\begin{Proposition}\label{Prop6}
Suppose $c_1(c^*-1) \leq 2$.
For every $\varepsilon>0$ and $\rho' \in (\rho^*,1]$, there exists $\underline{\delta} \in (0,1)$ such that for every $\pi \in \Delta(\mathcal{C})$ with $\pi_n \geq \varepsilon$ and $\delta>\underline{\delta}$, there exists an equilibrium in which the sender's payoff is $w(\rho')$.
\end{Proposition}
To understand how Proposition \ref{Prop6} completes the proof, notice that
\begin{itemize}
  \item $v_i(\rho^*) \geq w_i(\rho^*)$ for every $i \in \{1,2,...,n\}$ such that $c_1(c_i-1) \in [0, 2]$.
  \item $v_i(\rho^*) < w_i(\rho^*)$ for every $i \in \{1,2,...,n\}$ such that $c_1(c_i-1) >2$.
\end{itemize}
For every $i$ such that $c_i >1$, $w_i(\rho')$ is strictly increasing in $\rho'$ and $v_i(\rho)$ is strictly decreasing in $\rho$, with both functions being continuous. This suggests that
for every $\rho \in [0,\rho^*)$,
suppose
the sender attains payoff $v(\rho)$ in the auxiliary game, then
there exists $\rho' \in (\rho^*,1]$
such that
\begin{itemize}
  \item $w_i(\rho') < v_i(\rho)$ for every $i \in \{1,2,...,n\}$ with $c_1(c_i-1) \in [0, 2]$.
\end{itemize}
When the chosen $\rho$ and $\rho'$ in the previous step are both close enough to $\rho^*$,
\begin{itemize}
  \item $w_k(\rho') > v_k(\rho)$ for every $k \in \{1,2,...,n\}$ with $c_1(c_k-1) >2$.
\end{itemize}
Provided that $\delta$ is close enough to $1$,
\begin{itemize}
  \item every type with  $\overline{v}(c_i) \geq \underline{v}(c_i)$, i.e.,
 $c_1(c_i-1) \in [0, 2]$
 strictly prefers to follow the equilibrium strategy of type $c^*$
  \item every type with $\overline{v}(c_i) < \underline{v}(c_i)$ strictly prefers to play $\mathbf{a}^H$ in the active learning phase and strictly prefers to send message $l$ at the first history where the sender has never lied before and the receivers have ruled out types in $\underline{\mathcal{C}}$.
\end{itemize}
\begin{proof}[Proof of Proposition 6:]
I keep track of the following state variables:
\begin{itemize}
  \item[1.] The probability of type $c_n$ sender in the receiver's posterior belief, denoted by $\zeta(h^t)$. I call this the sender's \textit{reputation} at $h^t$.
  \item[2.] The promised continuation value to the sender, denoted by $v(h^t) \in \mathbb{R}^n$. Given that it is a convex combination of $v^H$, $v^N$ and $v^O$, it is equivalent to keep track of the convex weights of $v^H$, $v^N$ and $v^O$ in $v(h^t)$, which I denote by $p^H(h^t)$, $p^N(h^t)$, and $p^O(h^t)$, respectively.
  \item[3.] The highest-cost type in the support of receivers' posterior belief, denoted by $\overline{c}(h^t)$, as well as
  its probability according to the receiver's belief at $h^t$.
\end{itemize}
The initial values of these state variables are $\zeta(h^0)=\pi_n$, $p^H(h^0)=\frac{\rho' (1-p_h) c_1}{p_h +\rho' (1-p_h) c_1}$, $p^O(h^0)=\frac{(1-\rho') p_h}{p_h +\rho (1-p_h) c_1}$, $p^N(h^0) = \frac{\rho' p_h}{p_h +\rho' (1-p_h) c_1}$, $\overline{c}(h^0)=c_1$, and the probability of type $\overline{c}(h^0)$ is $\pi_1$. Recall the definition of the sender's effective discount factor:
\begin{equation*}
    \widehat{\delta} \equiv \delta \frac{1-p_h}{1-\delta p_h}.
\end{equation*}
Let $\zeta^* \in (0,\pi_n)$ and $\lambda>0$ be constants, defined similarly as the proof of Proposition \ref{Prop4}. I partition the set of on-path histories into three subsets, depending on the value of $p^O(h^t)$.
\begin{itemize}
  \item Class 1 histories: $p^O(h^t) \geq 1-\widehat{\delta}$,
  \item Class 2 histories: $p^O(h^t) \in (0, 1-\widehat{\delta})$,
  \item Class 3 histories: $p^O(h^t) =0$.
\end{itemize}

\paragraph{Class 1 Histories:} For every $h^t$ such that $p^O(h^t) \geq 1-\widehat{\delta}$,
\begin{itemize}
  \item The receiver plays $\mathbf{b}^N$.
  \item All types of sender only play $\mathbf{a}^H$ and $\mathbf{a}^L$ with positive probability.
  The sender of types $c_1$ to $c_{n-1}$ play the same action, and type $c_n$ plays differently, with probabilities pinned down by:
  \begin{equation}\label{C.4}
    \zeta (h^t,(l,h))-\zeta^*=\min \{ \big(1+\lambda (1-\rho^*)\big) (\zeta(h^t)-\zeta^*),1-\zeta^* \},
  \end{equation}
  \begin{equation}\label{C.5}
    \zeta (h^t,(l,l))-\zeta^*=\big(1-\lambda \rho^*\big)(\zeta(h^t)-\zeta^*),
  \end{equation}
  and $\zeta(h^t,(h,h))=\zeta(h^t,(h,l))=\zeta(h^t)$.
\end{itemize}
One can verify that according to the receiver's belief, $\mathbf{a}^L$ is played with probability at least $\rho^*$, and therefore, she has an incentive to play $\mathbf{b}^N$. The sender's continuation value is defined recursively. In particular, $v(h^t,(h,h))=v(h^t,(h,l))=v(h^t)$,
\begin{equation}\label{C.6}
    v(h^t,(l,l)) = \frac{p^H(h^t)}{\widehat{\delta}}v^H
    + \frac{p^N(h^t)}{\widehat{\delta}}v^N
    + \frac{p^O(h^t)-(1-\widehat{\delta})}{\widehat{\delta}}v^O,
\end{equation}
and if $\zeta(h^t,(l,h))<1$, then
\begin{equation}\label{C.7}
    v(h^t,(l,h)) = \frac{p^H(h^t)}{\widehat{\delta}}v^H
    + \frac{p^N(h^t)-(1-\widehat{\delta})}{\widehat{\delta}}v^N
    + \frac{p^O(h^t)}{\widehat{\delta}}v^O;
\end{equation}
if $\zeta(h^t,(l,h))=1$, then
\begin{equation}\label{C.8}
    v(h^t,(l,h)) = q(h^t) v^H +(1-q(h^t)) v^N,
\end{equation}
where $q(h^t)$ is pinned down by:
\begin{equation}\label{C.9}
    q(h^t) p_h -(1-q(h^t)) c_n(1-p_h)
    =\frac{p^H(h^t)}{\widehat{\delta}} p_h - \frac{p^N(h^t)-(1-\widehat{\delta})}{\widehat{\delta}} c_n(1-p_h).
\end{equation}
\paragraph{Class 2 Histories:} For every $h^t$ such that $p^O(h^t) \in (0, 1-\widehat{\delta})$, the receiver plays $\mathbf{b}^N$.
All types of sender in the support of receiver's belief at $h^t$ plays $\mathbf{a}^L$ except for type $\overline{c}(h^t)$, who mixes between $\mathbf{a}^H$ and $\mathbf{a}^L$. To specify type $\overline{c}(h^t)$'s
 mixing probabilities,
let
\begin{equation}\label{C.10}
    l(h^t) \equiv \# \Big\{
    h^s \Big| h^s\prec h^t, h^s \textrm{ belongs to Class 2}, \omega_{s}=l, \textrm{ and } \overline{c}(h^s)=\overline{c}(h^t)
    \Big\}
\end{equation}
be the number of histories that (1) strictly precede $h^t$, and (2) the highest lying cost type in the support of the receiver's belief is
$\overline{c}(h^t)$, and (3) the state realized at history $h^s$ is $l$.
\begin{enumerate}
  \item If $\overline{c}(h^t)=c_j$ with $j \leq n-2$, then
type $\overline{c}(h^t)$ plays $\mathbf{a}^H$ at $h^t$ with probability
\begin{equation}\label{C.11}
\frac{1}{k_j-l(h^t)}
\end{equation}
and $\mathbf{a}^L$ with complementary probability, with
$k_j$ being the smallest integer $k \in \mathbb{N}$ such that:
\begin{equation}\label{C.12}
    \frac{\pi_j/k}{\pi_j/k +\pi_{j+1}+...+\pi_n} \leq 1-\rho^*.
\end{equation}
  \item If $\overline{c}(h^t)=c_{n-1}$, then
type $c_{n-1}$ sender plays $\mathbf{a}^H$ at $h^t$ with probability
\begin{equation}\label{C.13}
    \min\{1, \frac{1-\rho^*}{1-\zeta(h^t)} \},
\end{equation}
and $\mathbf{a}^L$ with complementary probability.
\end{enumerate}
The sender's continuation value is given by $v(h^t,(h,h))=v(h^t)$, $v(h^t,(h,l))=v(h^t)$,
\begin{equation}\label{C.14}
v(h^t,(l,l)) \equiv  Q(h^t) v^H +(1-Q(h^t)) v^N,
  \end{equation}
  where $Q(h^t)$ is pinned down by the following equation:
  \begin{equation}\label{C.15}
    \frac{p^H(h^t)}{\widehat{\delta}} p_h
    -\frac{p^N(h^t)}{\widehat{\delta}} \overline{c}(h^t) (1-p_h)
    =Q(h^t) p_h -(1-Q(h^t)) \overline{c}(h^t) (1-p_h).
  \end{equation}
Therefore, type $\overline{c}(h^t)$ sender weakly prefers to send message $l$ when the realized state is $l$, while other types of senders in the support of receiver's belief at $h^t$ strictly prefers to send message $h$ when the realized state is $l$.
The sender's continuation value after sending message $h$ when $\omega_t=l$
   depends on whether $\zeta(h^t,(l,h))$ equals $1$ or not, with
$\zeta(h^t,(l,l))$ computed via Bayes Rule given the receiver's belief at $h^t$ and type $\overline{c}(h^t)$'s mixing probability at $h^t$:
\begin{enumerate}
  \item If $\zeta(h^t,(l,h))<1$, then the sender's continuation payoff at $(h^t,(l,h))$ is:
  \begin{equation*}
        v(h^t,(l,h))= \frac{p^H(h^t)}{\widehat{\delta}} v^H
    +\frac{p^O(h^t)}{\widehat{\delta}} v^O
    +\frac{p^N(h^t)-(1-\widehat{\delta})}{\widehat{\delta}} v^N.
  \end{equation*}
  \item If $\zeta(h^t,(l,h))=1$, then the sender's continuation payoff at $(h^t,(l,h))$ is:
\begin{equation*}
    v(h^t,(l,h)) = q(h^t) v^H +(1-q(h^t)) v^N,
\end{equation*}
where $q(h^t)$ is pinned down by:
\begin{equation*}
    q(h^t) p_h -(1-q(h^t)) c_n(1-p_h)
    =\frac{p^H(h^t)}{\widehat{\delta}} p_h - \frac{p^N(h^t)-(1-\widehat{\delta})}{\widehat{\delta}} c_n(1-p_h).
\end{equation*}
\end{enumerate}

\paragraph{Class 3 Histories:} For every $h^t$ such that $p^O(h^t) =0$, players' strategies are the same as in the absorbing phase in the proof of Proposition \ref{Prop4}. Namely, learning stops on the equilibrium path, all types of the sender in the support of receiver's belief plays the same action, the continuation play on the equilibrium path consists only of outcomes $(\mathbf{a}^H,\mathbf{b}^T)$
and $(\mathbf{a}^L,\mathbf{b}^N)$. Moreover, all types in the support of receiver's belief receives payoff at least zero. Such a construction is presented in Appendix A,
which I omit to avoid repetition.

\paragraph{Promise Keeping Constraint:} Verifying the promise keeping constraint uses the same argument as the proof of Theorem 1 in Pei (2019). In particular, recall the definitions of \textit{reduced outcome path} in Appendix \ref{secA}. The formulas for posterior beliefs in (\ref{C.4}) and (\ref{C.5}) suggests a result similar to Lemma A.1 in Pei (2019) that along every reduced outcome path, for example from $h^0$ to $h^t$, if the discounted average frequency of message $h$ divided by the discounted average frequency of message $l$ exceeds $\frac{\rho^*}{1-\rho^*}$, then play will reach a Class 2 or Class 3 history before play reaching $h^t$. Using this lemma and the sender's strategy at Class 2 histories, one can establish that along every reduced outcome path, the number of periods in which play stays at Class 2 histories (while the realized state is $l$) is uniformly bounded from above. As $t \rightarrow \infty$, play reaches a Class 3 history with probability $1$ and the continuation payoff at those histories can be delivered via an equilibrium in which learning about the sender's lying cost does not happen on the equilibrium path.

\end{proof}
\section{Consequentialism View on Lying}\label{secD}
My baseline model embodies the view that the sender suffers a psychological cost of lying regardless of the harm it causes on the receiver. I adopt a different view in this section, in which the sender suffers from the lying costs if and only if he is trusted by the receiver and that his lie has caused negative payoff consequences to the receiver. Formally, let $\boldsymbol{\beta}: M \rightarrow \Delta (A)$ be the sender's belief about receiver's reaction to his messages. The sender's stage-game lying cost is:
\begin{equation}\label{4.1}
    c \cdot \mathbf{1}\{m \neq \omega\} \cdot
    \Big(
    \max_{m' \in M} u_r(\omega, \boldsymbol{\beta}(m'))
    -u_r(\omega, \boldsymbol{\beta}(m))
    \Big).
\end{equation}
Players' stage-game payoffs under the four pure action profiles defined in section \ref{sec3} are given by the following matrix:
\begin{center}
\begin{tabular}{| c | c | c |}
  \hline
  $-$ & $\mathbf{b}^T$ & $\mathbf{b}^N$ \\
  \hline
  $\mathbf{a}^H$ & $p_h,p_h$ & $0,0$ \\
  \hline
  $\mathbf{a}^L$ & $p_h+(1-c)(1-p_h),2p_h-1$ & $0,0$ \\
  \hline
\end{tabular}
\end{center}
For every $j \in \{1,2,...,n\}$, let
\begin{equation}\label{4.2}
    v_j^{\dagger} \equiv p_h \frac{2-c_j}{2-c_1}.
\end{equation}
Note that $v_1^{\dagger}=p_h=v_1^*$, $v_j^*> v_j^{\dagger}>p_h$ for every $j \geq 2$.
Let $v^{\dagger} \equiv (v_1^{\dagger},...,v_n^{\dagger})$.
Player $1$'s highest equilibrium payoff is characterized in the following theorem:
\begin{Theorem1}
There is no BNE such that type $c_1$ attains payoff strictly more than $v_1^{\dagger}$. For every $\varepsilon>0$, there exists $\underline{\delta} \in (0,1)$ such that when $\delta >\underline{\delta}$:
\begin{enumerate}
\item There is no BNE such that type $c_j$ attains payoff more than $v_j^{\dagger}+\varepsilon$ for some $j \in \{2,3,...,n\}$.
\item There exists a sequential equilibrium in which the sender attains payoff within $\varepsilon$ of $v^{\dagger}$.
\end{enumerate}
\end{Theorem1}
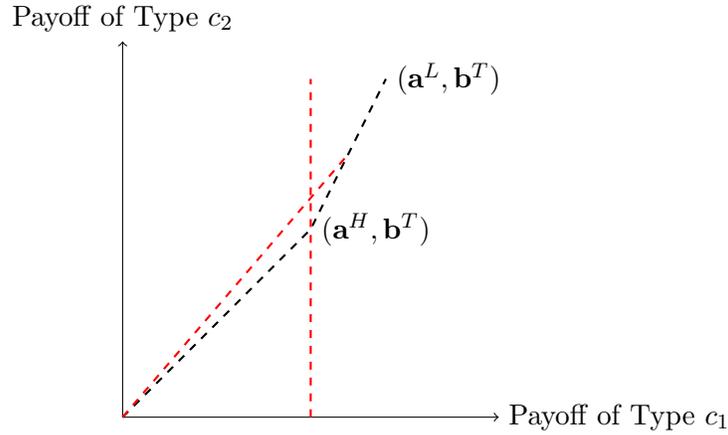
\begin{figure}\label{Figure2}
\begin{center}
\begin{tikzpicture}[scale=0.5]
\draw [->] (0,0)--(10,0)node[right]{Payoff of Type $c_1$};
\draw [->] (0,0)--(0,10)node[above]{Payoff of Type $c_2$};
\draw [dashed, thick] (0,0)--(5,5)node[right]{$(\mathbf{a}^H,\mathbf{b}^T)$}--(7,9)node[right]{$(\mathbf{a}^L,\mathbf{b}^T)$};
\draw [dashed, red, thick] (0,0)--(6,7);
\draw [dashed, red, thick] (5,0)--(5,9);
\end{tikzpicture}
\caption{Intersection between the two red lines is the sender's highest equilibrium payoff under consequentialism lying cost.}
\end{center}
\end{figure}
The highest payoff vector $(v_1^{\dagger},...,v_n^{\dagger})$ is depicted in Figure 2.
The proof of the necessity part is similar to that of Propositions \ref{Prop1}, \ref{Prop2} and \ref{Prop3}, which I omit to avoid repetition.
Since
players' stage-game payoffs are monotone-supermodular, the sufficiency part of the proof is similar to the proof of Theorem 1 in Pei (2019), with the exception that $\delta$ needs to be replaced by a strictly lower number:
\begin{equation}\label{4.3}
    \widehat{\delta} \equiv \delta \frac{1-p_h}{1-\delta p_h},
\end{equation}
and the receiver's belief as well as the sender's continuation value change in the active learning phase only when the realized state is $l$ but not when the realized state is $h$.
My result on the sender's equilibrium behavior in equilibria that are approximately optimal for the sender also generalizes, which is stated as Theorem 2'.
\begin{Theorem2}
Suppose $n \geq 2$. For every small enough $\varepsilon>0$, there exists $\underline{\delta} \in (0,1)$ such that when $\delta >\underline{\delta}$, for every equilibrium in which the sender attains payoff within $\varepsilon$ of $v^{\dagger}$, no type of the sender plays $\mathbf{a}^H$ and $\mathbf{a}^L$ with positive probability at all on-path histories.
\end{Theorem2}
\begin{proof}[Proof of Theorem 2':] Suppose toward a contradiction that in some BNE $\sigma$ that attains payoff within $\varepsilon$ of $v^{\dagger}$, and
there exists a type $c_j$ that has a completely mixed best reply against $\sigma_r$. Therefore, playing $\mathbf{a}^H$ at every on-path history and playing $\mathbf{a}^L$ at every on-path history are both his best replies against $\sigma_r$. Theorem 1 in Liu and Pei (2020) implies that:
\begin{itemize}
  \item For every $i<j$, type $c_i$ plays $\mathbf{a}^H$ with probability $1$ at every on-path history.\footnote{Different from the binary action game studied in Pei (2019), it is \textit{not} true that for every $k<j$, type $c_k$ plays $\mathbf{a}^L$ with probability $1$ at every on-path history. This is because the sender has other stage-game actions, such as lying in every state, in which case he suffers strictly higher cost of lying compared to $\mathbf{a}^L$.}
\end{itemize}
I consider two cases separately.
First, if $j \geq 2$, then type $c_1$ plays $\mathbf{a}^H$ with probability $1$ at every on-path history. For type $c_2$, he fully separates from type $c_1$ the first time he sends message $h$ when the state is $l$, after which he is the type that has the highest lying cost and his continuation payoff is
no more than $p_h$. As a result, his payoff in period $0$ is no more than $(1-\delta)+\delta p_h$, which
is strictly lower than $v_2^{\dagger}$ as $\delta \rightarrow 1$. This leads to a contradiction.

Second, if $j=1$, then type $c_1$ finds it optimal to play $\mathbf{a}^H$ in every period, and is also optimal to play $\mathbf{a}^L$ in every period. Since the sender's equilibrium payoff is within $\varepsilon$ of $v^{\dagger}$, type $c_1$'s payoff is at least $v_1^*-\varepsilon$ by playing $\mathbf{a}^L$ in every period, and type $c_2$'s payoff from doing so is no more than $v_2^*+\varepsilon$.
Since $p_h<1/2$, the receiver's stage-game action of playing $H$ following every message is strictly suboptimal, and cannot be played at any on-path history. Among the remaining three stage-game actions of the receiver's, the sender's stage-game payoff is $1-(1-p_h)c$ under $(\mathbf{a}^L,\mathbf{b}^T)$, $0$
under $(\mathbf{a}^L,\mathbf{b}^N)$.
Let $Q_L$ be the occupation measure of outcome
$(\mathbf{a}^L,\mathbf{b}^T)$ when the sender plays $\mathbf{a}^L$ in every period and the receiver plays according to $\sigma_r$.
Type $c_1$'s payoff by playing $\mathbf{a}^L$ in every period equals:
\begin{equation*}
    Q_L \big( 1-(1-p_h)c_1 \big),
\end{equation*}
which is no less than $p_h-\varepsilon$ since playing $\mathbf{a}^L$ in every period is his equilibrium best reply, and his equilibrium payoff is at least $p_h-\varepsilon$. This gives:
\begin{equation*}
    Q_L \geq \frac{p_h-\varepsilon}{1-(1-p_h)c_1}.
\end{equation*}
By playing $\mathbf{a}^L$ in every period, type $c_2$'s payoff is:
\begin{equation*}
    Q_L \big( 1-(1-p_h)c_2 \big).
\end{equation*}
The inequality on $Q_L$ leads to a lower bound on the above expression, which is
$(p_h-\varepsilon) \frac{1-(1-p_h)c_2}{1-(1-p_h)c_1}$.
Given that $p_h<1/2$, this is strictly greater than $v_2^{\dagger}+\varepsilon$ when $\varepsilon$ is small enough.
This contradicts the conclusion of Theorem 1' that type $c_2$'s equilibrium payoff when $\delta$ is large enough cannot exceed $v_2^{\dagger}+\varepsilon$.
\end{proof}

\end{spacing}

\end{document}